\documentclass[pra,aps,floatfix,superscriptaddress,11pt,tightenlines,longbibliography,onecolumn,notitlepage]{quantumarticle}

\usepackage{graphicx}

\usepackage{amsmath, amsfonts, amssymb, amsthm, braket, bbm, xcolor}
\usepackage{subfigure}
\usepackage{tikz}
\usepackage{comment}
\usepackage[sort&compress,numbers,square]{natbib}

\definecolor{darkblue}{rgb}{0.,0.,0.4}
\definecolor{darkred}{rgb}{0.5,0.,0.}
\usepackage[pdftex,colorlinks=true,linkcolor=darkblue,citecolor=darkred,urlcolor=darkblue]{hyperref}
\hypersetup{pdftitle={Limits on the storage of quantum information in a volume of space}, pdfauthor={Steven T. Flammia, Jeongwan Haah, Michael J. Kastoryano, Isaac H. Kim}}

\newcommand{\fidelity}{\mathfrak{F}}
\newcommand{\bures}{\mathfrak{B}}
\newcommand{\trdist}{\mathfrak{T}}
\renewcommand{\Re}{\mathfrak{Re}}
\newcommand{\onenorm}[1]{\left\| {#1} \right\|_1}
\newcommand{\norm}[1]{\left\| {#1} \right\|}
\newcommand{\twonorm}[1]{\left\| {#1} \right\|_2}

\newcommand{\abs}[1]{\left|#1\right|}
\DeclareMathOperator{\Tr}{\mathrm{Tr}}

\newcommand{\id}{{\mathrm id}}

\newcommand{\cN}{{\mathcal N}}
\newcommand{\cR}{{\mathcal R}}

\newcommand{\cM}{{\mathcal M}}
\newcommand{\cS}{{\mathcal S}}
\newcommand{\cV}{{\mathcal V}}

\newcommand{\half}{\frac{1}{2}}

\newcommand{\be}{\begin{equation}}
\newcommand{\ee}{\end{equation}}
\newcommand{\bea}{\begin{eqnarray}}
\newcommand{\eea}{\end{eqnarray}}

\newtheorem{lem}{Lemma}
\newtheorem{thm}[lem]{Theorem}
\newtheorem{defi}[lem]{Definition}
\newtheorem{corollary}[lem]{Corollary}

\newtheorem{result}{Result}

\begin{document}

\title{Limits on the storage of quantum information in a volume of space}
\date{}

\author{Steven T.\ Flammia}
\affiliation{Centre for Engineered Quantum Systems, School of Physics, The University of Sydney, Australia}
\affiliation{Center for Theoretical Physics, Massachusetts Institute of Technology, Cambridge, USA}
\author{Jeongwan Haah}
\affiliation{Station Q Quantum Architectures and Computation Group, Microsoft Research, Redmond, Washington, USA}
\email{jwhaah@microsoft.com}
\affiliation{Center for Theoretical Physics, Massachusetts Institute of Technology, Cambridge, USA}
\author{Michael J.\ Kastoryano}
\affiliation{NBIA, Niels Bohr Institute, University of Copenhagen, Denmark}
\author{Isaac H.\ Kim}
\affiliation{IBM T. J. Watson Research Center, Yorktown Heights, New York, USA}
\affiliation{Perimeter Institute for Theoretical Physics, Waterloo ON N2L 2Y5, Canada}
\affiliation{Institute for Quantum Computing, University of Waterloo, Waterloo ON N2L 3G1, Canada}

\begin{abstract}
We study the fundamental limits on the reliable storage of quantum information in lattices of qubits
by deriving tradeoff bounds for approximate quantum error correcting codes.
We introduce a notion of local approximate correctability and code distance,
and give a number of equivalent formulations thereof,
generalizing various exact error-correction criteria.
Our tradeoff bounds relate
the number of physical qubits $n$,
the number of encoded qubits $k$,
the code distance $d$,
the accuracy parameter $\delta$ that quantifies how well the erasure channel can be reversed,
and the locality parameter $\ell$ that specifies the length scale at which the recovery operation can be done.
In a regime where the recovery is successful to accuracy $\delta$ that is exponentially small in $\ell$,
which is the case for perturbations of local commuting projector codes,
our bound reads $kd^{\frac{2}{D-1}} \le O\bigl(n (\log n)^{\frac{2D}{D-1}} \bigr)$ for codes on $D$-dimensional lattices of Euclidean metric.
We also find that the code distance of any local approximate code cannot exceed $O\bigl(\ell n^{(D-1)/D}\bigr)$
if $\delta \le O(\ell n^{-1/D})$.
As a corollary of our formulation of correctability in terms of logical operator avoidance,
we show that the code distance $d$ and the size $\tilde d$ of a minimal region that can support all approximate logical operators
satisfies $\tilde d d^{\frac{1}{D-1}}\le O\bigl( n \ell^{\frac{D}{D-1}} \bigr)$,
where the logical operators are accurate up to $O\bigl( ( n \delta / d )^{1/2}\bigr)$ in operator norm.
Finally, we prove that for two-dimensional systems if logical operators can be approximated
by operators supported on constant-width flexible strings,
then the dimension of the code space must be bounded.
This supports one of the assumptions of algebraic anyon theories,
that there exist only finitely many anyon types.
\end{abstract}

\maketitle

\section{Introduction}
Quantum information is susceptible to decoherence,
but the effect can be mitigated by redundantly encoding the information
in a quantum error correcting code.
Since a reliable qubit is a scarce resource,
it is desirable to achieve maximal protection with a minimum effort,
and one of the most promising approaches for achieving this
is to incorporate geometric locality into the structure of the code.

Local quantum error correcting codes
have been thoroughly studied over the past decade~\cite{Terhal2015},
with most of the work focusing on (topological) stabilizer or subsystem codes.
A quantum code is said to be \emph{local}
if its stabilizer or gauge generators are supported
on a geometrically local and bounded region of a lattice embeddable in some metric space.
These codes, like most classical and quantum error correcting codes,
are often characterized by three numbers $[[n,k,d]]$:
$n$ is the number of physical, error-prone qubits comprising the code,
$k$ is the maximal number of logical qubits that can be reliably encoded,
and $d$ is the distance of the code,
i.e.\ the minimum number of qubits that must be modified
to perform a nontrivial logical operation.

Not all values of the triple $[[n,k,d]]$ are achievable,
and locality in particular imposes additional constraints.
In addition to the constraints imposed by local codes, which we review below,
the no-cloning bound implies that exact quantum error-correcting codes
cannot correct $n/4$ arbitrary single-qudit errors~\cite{gottesman2009}.
In its smallest instance,
the bound  shows that there is no code on four physical qubits
that exactly corrects an arbitrary single-qudit error.

However, the main goal of quantum error correction is to \emph{reduce} effective error rates, and it suffices to only approximately correct errors for this purpose, provided that the approximation error is sufficiently low.
Interestingly, the no-cloning bound breaks down when one considers
such \emph{approximate} quantum error correcting codes.
Leung \emph{et al}.~\cite{leung1997} have shown this explicitly by constructing
a four-qubit code that approximately corrects an arbitrary single-qubit amplitude damping error.
This effect was demonstrated even more dramatically by Cr\'{e}peau \emph{et al}.~\cite{crepeau2005},
who constructed approximate codes that can approximately correct $\lfloor (n-1)/2\rfloor$
arbitrary single qudit errors, with an approximation error that is exponentially small in $n$.
This demonstrates that approximate codes can radically outperform exact codes by some measures.

Many local quantum codes can be represented
as ground spaces of gapped local Hamiltonians.
Often these systems have an exact unfrustrated ground space,
but there are much more general systems where we expect code properties to hold with some suitable notion of approximation.
Known examples include fractional quantum Hall states~\cite{moore1991}
and Kiteav's honeycomb model in the gapped phase~\cite{kitaev2006}.
Might these more general systems allow for unique possibilities for storing quantum information?
Many aspects of the unfrustrated Hamiltonian quantum codes remain stable under small perturbations,
such as the energy splitting in the ground space and the gap~\cite{thespiros}.
It is also known that the logical operators, once perturbed,
become dressed operators that are quasi-local~\cite{hastings2005, spirosbravyi2010}.
However, in light of the examples of Refs.~\cite{leung1997, crepeau2005},
the stability of all code properties for general systems deserves a careful and thorough examination.

Motivated by these observations, we initiate the study of
\emph{local approximate quantum error-correcting codes} (local AQEC).
It is one of the aims of the present paper
to establish a framework allowing the analysis of code properties of subspaces
in sufficient generality to encompass many of the interesting gapped many-body models,
not necessarily represented by commuting projector Hamiltonians or unfrustrated ground spaces.
Our contributions are threefold.

In the first part, we comprehensively analyze the abstract theory of AQEC codes.
We introduce a number of different notions of approximate quantum error correction
and study the relations between them.
In particular, we propose a notion of locally correctable codes
for which the class of local commuting projector codes are a subclass.
The different notions of correctability that we introduce are all exactly equivalent
when the approximation parameter $\delta$ vanishes,
and the relations we establish show that different notions
in fact have varying levels of robustness once $\delta > 0$.

In the second part, we use these relations
to generalize the existing tradeoff bounds for the parameters of local commuting projector codes.
These results include the parameter bounds
derived by Bravyi, Poulin, and Terhal~\cite{BPT}
and the constraints on the structure of the logical operator
derived by Haah and Preskill~\cite{HaahPreskill2012tradeoff} as special cases.
Our bounds are only slightly weaker than the prior results for exact codes,
implying that local AQEC cannot significantly outperform exact codes by these measures.
We furthermore establish a condition under which one of the key axioms of
(2+1)-dimensional topological quantum field theory can be derived:
If the logical operators can be continuously deformed away from any disk-like region,
the ground state degeneracy is bounded by a constant independent of the system size.
Roughly, the number of anyon species in 2D is finite.

Our proof techniques make significant departures from much of the previous literature on local codes
in that we make heavy use of information-theoretic ideas and methods.
We provide information-theoretic generalizations of well known tools from the theory of local codes,
including more general forms of the Union Lemma and the Cleaning Lemma.
Our proofs and definitions do not require the codes to have local generators;
rather all of our analysis is done at the level of subspaces.
Locality in only invoked at the level of the recovery operations.

By construction, our results are trivially applicable to local commuting projector codes,
but more importantly, every step in our derivations is applicable to the perturbation of such codes 
so long as the perturbed Hamiltonian remains in the same gapped phase. 
This class includes explicit models such as the Kitaev's honeycomb model~\cite{kitaev2006},
the [4,2,2]-concatenated toric code~\cite{Brell2011, Criger2016}, 
and codes based on entanglement renormalization 
that generate holographic quantum codes that are only approximate codes~\cite{Kim2017}.
Thus, we establish that error correction in realistic Hamiltonian systems
that only approximate local commuting projector codes is indeed robust to small imperfections.

\subsection*{Prior Work}

Previous work on exact quantum codes
has placed limits on the achievable code parameters.
Beyond the no-cloning bound~\cite{gottesman2009} mentioned above
that holds for general exact codes,
locality imposes further constraints on the code parameters.
Bravyi and Terhal~\cite{BravyiTerhal2009} have shown that stabilizer and subsystem codes in $D$ spatial dimensions
obey the bound $d \le O(L^{D-1})$ where $n=L^D$ is the number of physical qubits in the code in a Euclidean lattice.
The result on stabilizer codes was extended to the larger class of local commuting projector codes
by Bravyi, Poulin, and Terhal~\cite{BPT},
who also showed that $k d^{2/(D-1)}\leq O(n)$.
Delfosse~\cite{Delfosse2013a} showed that these arguments can also be adapted to hyperbolic lattices, and showed that surface codes and color codes (with $D=2$) satisfy $k d^2\leq O\bigl(n (\log n)^2\bigr)$, and that this is scaling is achievable.
The result on subsystem codes was adapted by Bravyi~\cite{Bravyi2011b}
to show that subsystem codes satisfy the bound $kd^{1/(D-1)}\leq O(n)$.
Quantum stabilizer codes based on the toric code saturate these bounds for $D=2$,
and subsystem codes that saturate for $D=2$~\cite{Bravyi2011b} or nearly saturate for $D\ge3$~\cite{Bacon2015} are also known.

For the class of AQEC, no prior work has shown nontrivial tradeoff bounds like the ones above,
and most of the relevant work (such as Ref.~\cite{crepeau2005} above)
is confined to codes without any evident locality structure.
Known examples of AQEC have generally all been applied to the study of amplitude damping channels
(see e.g.~\cite{leung1997, Cafaro2014, Ouyang2014, Grassl2015}).
Conditions for approximate correction have been derived in terms of the coherent information~\cite{Schumacher2002}
as well as in terms of error metrics that relate to the Knill-Laflamme conditions
for exact error correction~\cite{BenyOreshkov2010, NgMandayam2010}.
Some progress has also been made on finding explicit efficient representations of logical operators for approximate codes~\cite{Bridgeman2016, Chubb2016}.

\subsection*{Overview of results}

Our results depend on defining appropriate measures of locality and approximation.
We will precisely define the relevant measures in Definitions~\ref{def:correctability} and~\ref{def:localCorrectability},
but, roughly speaking, $\ell$ will denote a length scale (in lattice spacing units) for the action of a recovery map for erasure errors, and $\delta$ will denote the accuracy of the recovery.

Our first contribution is to relate five operationally distinct notions of correctability, decoupling, and cleaning.
We will establish similarities between these notions in the approximate recovery setting by deriving inequalities among them without any dimensional factors of the Hilbert spaces involved.
When the approximation parameter vanishes, these notions become exactly equivalent.

Decompose the lattice into disjoint regions as $\Lambda=ABC$ and let $R$ denote a purifying system
(see e.g.~Fig.~\ref{fig:decoupling}, though there is no restriction on the geometry of the regions).
We denote an operator supported on a subsystem using superscripts;
for example, $\rho^{AB}$ is a state supported on $AB$,
and $\rho^A$ is its partial trace supported on $A$.
Then we have our first main result.


\begin{result}[Imprecise formulation] The following statements are equivalent.
\begin{itemize}

\item[\emph{(I)}] There exists a local quantum channel $\cR_{B}^{AB}$ with support on $AB$
approximately recovering the erasure of region $A$: $\cR_{B}^{AB}(\rho^{BCR})\approx \rho^{ABCR}$.

\item[\emph{(II)}] 

There exists a disentangling unitary $U^B$ on $B=B_1B_2$
that approximately turns any code state into a product state between $AB_1$ and $B_2C$:
$U^B\rho^{ABCR}U^{B \dag}\approx \omega^{AB_1} \otimes \rho^{B_2 CR}$.

\item[\emph{(III)}] Tracing out $B$ approximately decouples $A$ from $C$:
$\rho^{ACR}\approx\rho^A\otimes\rho^{CR}$.

\item[\emph{(IV)}] 
Any logical operator $U^{ABC}$ is associated with another operator $V^{BC}$
supported on the complement of $A$
that act equivalently when restricted to the code space $\Pi$: $U^{ABC}|_\Pi\approx V^{BC}|_\Pi$.

\end{itemize}
In particular, we have
\begin{align}\label{eq:eqs}
 \emph{\text{(I)}} =_\ell \emph{\text{(II)}} \sim_\ell \emph{\text{(III)}} \sim \emph{\text{(IV)}} ,
\end{align}
where `$\sim$' means that the magnitude of the error is not preserved in the implication,
and the subscript `$\ell$' means that the equivalence holds with locality.
\end{result}
\noindent
Here we have been imprecise about the exact nature of the equivalence,
and in particular we have not quantified or defined our notion of approximation or what it means for the equivalence to hold with locality.
These statements are given precise formulations and proofs beginning in Section~\ref{sec:I} as follows:
(I)~Definition~\ref{def:localCorrectability},
(II)~Corollary~\ref{cor:disent},
(III) Theorems~\ref{thm:infodist},~\ref{thm:decoupling},
 and 
(IV)~Theorems~\ref{thm:cleaning1},~\ref{thm:cleaning2}.

Our next main result shows that for AQECs that have local correctability
parametrized by an approximation parameter $\delta$ and a length scale $\ell$,
we have a tradeoff between the capacity and the reliability of encoded information.
Consider a Euclidean
$D$-dimensional lattice of linear size $L$
for which we intentionally do not specify any boundary conditions.
\begin{result}
For a $D$-dimensional AQEC with parameters $[[n,k,d,\delta,\ell]]$, it holds that
\begin{align}
\bigl(1 - c\, \epsilon \log(1/\epsilon)  \bigr) k d^{\frac{2}{D-1}} \leq  c' n \ell^{\frac{2D}{D-1}}\,,
\end{align}
where $\epsilon = n\delta/d$, and $c,c' > 0$ are absolute constants.
\end{result}
\noindent
In the case that $\delta = 0$ and $\ell$ is constant,
as is the case for local commuting projector codes,
we recover the tradeoff bounds of Bravyi, Poulin, and Terhal~\cite{BPT} as a special case.
Moreover, we show that for relevant examples
(like perturbed versions of local commuting projector codes)
the approximation error $\delta$ vanishes sufficiently quickly
that the inequality remains meaningful in spite of the $n$ dependence inside the parentheses. 
When the accuracy $\epsilon$ is exponentially small in $\ell$ we obtain the bound quoted in the abstract of $kd^{\frac{2}{D-1}} \le O\bigl(n (\log n)^{\frac{2D}{D-1}} \bigr)$. 
This result is restated and proven as Theorem~\ref{thm:ABPT} in Section~\ref{sec:tradeoff}.
Using the approximate cleaning lemma from Section~\ref{sec:I}
and similar methods as in the proof of Theorem~\ref{thm:ABPT}
we prove tradeoff bounds on the support of logical operators of the code
(a logical operator is one that preserves the code subspace).

\begin{result}
For a $D$-dimensional AQEC with parameters $[[n,k,d,\delta,\ell]]$
on a lattice of linear size $L$,
if $10L\delta<\ell$,
then the code distance is bounded from above by $5\ell L^{D-1}$.
Furthermore,  there exists a region $Y$ that contains $\tilde d$ qubits such that
every unitary logical operator $U$ can be approximated by an operator $V$ on $Y$
where
\begin{align}
& \norm{(U - V)\Pi} \le O\Bigl(  \sqrt{n\delta/d} \Bigr) \qquad \text{ and }
\qquad \tilde d d^{\frac{1}{D-1}}\le O\Bigl( n \ell^{\frac{D}{D-1}} \Bigr)\,.
\end{align}
\end{result}

Finally, we show how the framework of AQEC can be used
to derive one of the axioms of topological quantum field theory in (2+1) dimensions:
the ground space degeneracy is constant independent of system size.
We first define a notion of ``flexible'' logical operator that is made precise in Definition~\ref{def:flex},
which asserts roughly that logical operators can be deformed across the lattice
without changing much how they act on the code space.
We then derive the following result.
\begin{result}
For any code space $\Pi$ of a 2-dimensional system admitting flexible logical operators,
it holds that
\begin{align}
	\dim \Pi \le \exp\bigl(c\, \ell^2\bigr)\,,
\end{align}
where $c > 0$ is an absolute constant
and $\ell$ is the width of the strip that supports sufficiently faithful logical operators.
\end{result}
\noindent
The precise formulation and proof of this statement can be found in Section~\ref{sec:flex}.

\subsection*{Notation and conventions}

As mentioned above, 
$\Lambda$ is a regular $D$ dimensional lattice of side length $L$ 
without any specified boundary conditions.
We consider $\Lambda \subset \mathbb Z^D$ 
to be embedded in $\mathbb R^D$ with a Euclidean metric.
Each site is occupied with a qubit $\mathbb C^2$.
The assumption that a qubit, rather than a qudit, occupies each site
is not a restriction; 
a higher density of degrees of freedom can be accounted for
by rescaling the metric.
We use letters $A,B,C,D,E,X,Y,Z\in\Lambda$ to denote subsystems.
We use $\Pi$ to denote a subspace of the joint Hilbert space of these subsystems, which we identify as a code space.
Unless otherwise specified, different letters stand for disjoint subsystems.
We will often say, for example, that a code $\Pi$ is on $ABC$ when the system into which the code is embedded is divided into three disjoint subsystems.
The letter $R$ is reserved for a purifying space of the code space $\Pi$.
Therefore, as complex vector spaces, $\Pi$ and the Hilbert space of $R$ are isomorphic.
We will  use the same symbol $\Pi$ to denote the projection operator onto the code space.
Also, we will write $\rho \in \Pi$ to mean $\rho = \Pi \rho = \rho \Pi$.
If a state $\rho$ is pure, we will sometimes use $\ket \rho$ to denote a state vector.

When it is necessary to clarify the domain and the codomain of a linear operator,
we will use subscripts for domain, and superscripts for codomain.
For example, $V_B^{AE}$ denotes a linear operator from $B$ to $AE$.
The subscript will be omitted when the domain is equal to the codomain.
For example, $V^A$ is an operator on $A$.
(Under strict practice of our notation,
$V^A$ could mean a map from $\mathbb C$ to the Hilbert space of $A$,
which would specify a state vector, but we never use such a map.)
The same rule applies for quantum channels (i.e.\ completely positive and trace preserving maps).
For example, $\Tr_A$ is the erasure channel for system $A$, mapping from $A$ to scalars.
If a channel is tensored with the identity channel,
the identity component will be omitted.
For instance, we will write $\cR_B^{AD}(\sigma^{BC})$, which is a density operator on $ACD$,
in place of $(\cR_B^{AD} \otimes \id^C_C)(\sigma^{BC})$.
If $\rho^{AB}$ is clear from the context, we will simply write $\rho^B$ in place of $\Tr_A(\rho^{AB})$.

Noise will be modeled by quantum channels $\cN$.
It will be convenient to represent the noise in Stinespring form,
where $\cN(\rho)=\Tr_E(V \rho\otimes\varphi_E V^\dag)$
with $E$ denoting the purifying system (environment)
whose dimension can be taken as the dimension of the channel's input,
$\varphi_E$ is some (pure) state of the purifying system,
and $V$ is a unitary operator on the system and purifying space.
The complementary channel is obtained by tracing out the system:
$\hat{\cN}(\rho)=\Tr_S(V (\rho\otimes\varphi_E) V^\dag)$.
We will mainly consider erasure noise.
Specifically, $\cN$ is chosen to be $\mathcal{N}(\cdot) = \Tr_X(\cdot)$
for some subsystem $X$.

For any two density operators $\rho$ and $\sigma$,
the trace distance is defined as
\begin{align}
 \trdist(\rho,\sigma) = \frac 12 \onenorm{\rho - \sigma}.
\end{align}
We will mostly use the one-norm directly instead of the symbol $\trdist$.
The fidelity and Bures distance are defined as
\begin{align}
 \fidelity( \rho, \sigma ) &= \Tr\sqrt{\sqrt \sigma ~ \rho \sqrt{\sigma}} = \max_{\ket{\psi_\rho},\ket{\psi_\sigma}} \abs{\braket{\psi_\rho | \psi_\sigma}},\\
 \bures(\rho,\sigma ) &= \sqrt{ 1 - \fidelity(\rho,\sigma) },
\end{align}
where the second equality for the fidelity follows from Uhlmann's theorem~\cite{uhlmann1976}.
The Bures distance is a metric and satisfies the triangle inequality.
These quantities satisfy the Fuchs-van de Graaf relations~\cite{Fuchs1999}
\begin{align}
 \fidelity^2 + \trdist^2 &\le 1 \le \fidelity + \trdist,\\
 2\bures^2(\rho,\sigma) &\le \onenorm{\rho - \sigma} \le 2\sqrt{2} \bures(\rho,\sigma). \label{eq:FGrelation}
\end{align}

We will often encounter  optimizations of the form
\begin{align}
 \max_{\rho^{ABR}} \bures( \rho^{AR}, \rho^A \otimes \rho^R )
\end{align}
where a code $\Pi$ resides on $AB$.
In such an expression, $\rho^{ABR}$ denotes any density matrix of the code vectors possibly entangled with $R$.
Since our discussion will be only for finite dimensional Hilbert spaces,
the set of all density matrices is compact under the usual subspace topology inherited from $\mathbb R^N$,
which is the metric topology under any mentioned metric.
This means that the supremum is always attained at some state.

All $\log$ and $\exp$ functions have base $e \simeq 2.718$.

\section{Notions of Approximate Quantum Error Correction}\label{sec:I}

Most generally, an error correction scheme should be defined
in terms of encoding and decoding quantum algorithms~\cite{crepeau2005}.
This definition is so general that it includes schemes
where a superposition of physical states that are outputs of the encoding algorithm
may not correspond to a possible output of the encoding algorithm, e.g.~Ref.~\cite{crepeau2005}.
In this paper, we restrict ourselves to subspace codes.
Namely, we identify a code with a subspace of an $n$-qubit system,
and the encoding map is an isometry from $\mathbb C^{2^k}$ to $(\mathbb C^2)^{\otimes n}$,
where $k$ is the number of encoded qubits.

An exact quantum error correcting code has a distance $d$
if any measurement on fewer than $d$ qubits reveals no information about the encoded state.
In other words, a distance $d$ code admits a recovery map that can perfectly reverse the erasure of any $d-1$ qubits:
for any $M \subseteq \Lambda$ such that $|M|<d$,
there exists a completely-positive trace-preserving map $\mathcal{R}$ such that
\begin{align}
\cR \circ \Tr_M(\rho) = \rho \label{eqn:exactQEC}
\end{align}
for all states $\rho$ in the code subspace.
An equivalent formulation of the error correction condition,
known as the Knill-Laflamme condition~\cite{KnillLaflamme1997Theory},
expresses Eq.~\eqref{eqn:exactQEC} in terms of the Kraus operators of the noise channel.

As remarked earlier, the recovery map will never be perfect in practice
and does not have to be perfect;
rather, the recovery should be of high enough fidelity to suit the purpose of the code.
We use the best recovery map in terms of Bures distance to define approximate quantum error correction.
\begin{defi}\label{def:correctability}
We say that a region $A\subset\Lambda=AB$ is {\bf $\delta$-correctable} 
on $\Pi$ if there exists a recovery operation $\cR_B^{AB}$ such that
\begin{align}
 \bures\bigl(\cR_B^{AB}(\rho^{BR}), \rho^{ABR}\bigr)\leq \delta,
\end{align}
for any code state $\rho^{AB}\in\Pi$ and $\rho^{ABR}$ is any purification.
The {\bf code distance} of an approximate error correction code $\Pi$ is
defined as the largest integer $d$ such that
any region $A \subseteq \Lambda$ of size $|A| < d$ is $\delta$-correctable on $\Pi$.
\end{defi}
\noindent
The perfect quantum error correction condition corresponds to the $\delta=0$ case.
In our definition,
the code space $\Pi$ alone does not determine the accuracy paramter $\delta$ or the code distance $d$;
rather, the code space gives 
a function from $\delta$ to $d$.
Note that under this definition, the code distance is a non-decreasing function of $\delta$.

Our main intuition for local error correcting codes
comes from topological systems such as the toric code.
There, errors appear as excitation pairs (anyons),
that get annihilated when brought together.
The recovery of such errors therefore consists of bringing anyon pairs together.
In this picture, error recovery is a local operation.
This observation motivates the following notion
of a locally correctable region for general lattice models.

\begin{defi}
\label{def:localCorrectability}
We say that a region $A\subset\Lambda=ABC$,
with $AB=A^{+\ell}$ is the region that includes all qubits within a distance $\ell$ of $A$,
and $C$ is the complement of $AB$ in $\Lambda$, as in Fig.~\ref{fig:decoupling},
is {\bf $(\delta,\ell)$-correctable} on $\Pi$
if there exists a recovery operation $\cR_B^{AB}$ with support on $AB$ such that
\begin{align}
\bures\bigl(\cR_B^{AB}(\rho^{BCR}), \rho^{ABCR}\bigr)\leq \delta,
\end{align}
for any code state $\rho^{ABC} \in \Pi$ and $\rho^{ABCR}$ is any purification.
A code encoding $k$ qubits on $n$ physical qubits 
is said to have {\bf parameters} $[[n,k,d,\delta,\ell]]$
if every region of size less than $d$ is $(\delta,\ell)$-correctable.
\end{defi}
\noindent
By definition, any region that is $(\delta,\ell)$-correctable is $\delta$-correctable.

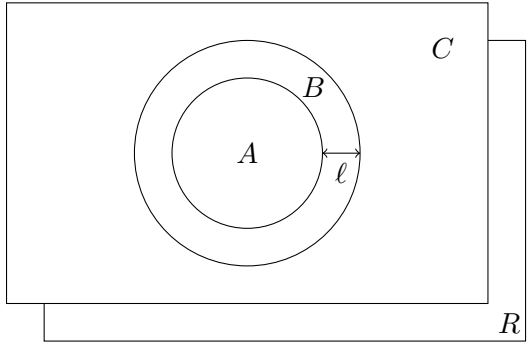
\begin{figure}
\centering
\begin{tikzpicture}
	\draw[shift={(.5,-.5)}] (-3.2,-2) rectangle (3.2,2);
	\draw[fill=white] (-3.2,-2) rectangle (3.2,2);
	\draw (0,0) circle (1);
	\draw (0,0) circle (1.5);
	\node at (0,0) {$A$};
	\node at (.88,.88) {$B$};
	\node at (2.6,1.4) {$C$};
	\node[below right] at (3.2,-2) {$R$};
	\draw[<->] (1,0) -- (1.5,0);
	\node[below, align=center] at (1.25,0) {$\ell$};
\end{tikzpicture}
    \caption{Decomposition of the lattice in the definitions of local approximate quantum error correction. $B$ shields $A$ from $C$ by a distance at least $\ell$, and $R$ represents the purifying space of the code. }
    \label{fig:decoupling}
\end{figure}

One of the central insights of quantum error correction is the information-disturbance tradeoff:
the existence of a recovery map implies that the environment knows almost nothing about the state, and vice versa.
A very general information disturbance tradeoff was derived in Ref.~\cite{BenyOreshkov2010}.
See also Refs.~\cite{kretschmann2008,NgMandayam2010,hayden2012}.
Here, we adapt it to our setting by assuming that the noise operation
is the partial trace over a given region $A$.
We further sharpen the statement to accommodate for local recovery.

\begin{thm}[Information-Disturbance tradeoff]\label{thm:infodist}
Let $A\subset\Lambda=ABC$, as in Fig.~\ref{fig:decoupling}, and define the constant
\begin{align}
\delta_\ell(A):=\min_{\omega^A}\sup_{\rho^{ABCR}}\bures\bigl(\omega^A\otimes\rho^{CR},~\rho^{ACR}\bigr) \label{def:BlA},
\end{align}
then
\begin{align}
\inf_{\cR_B^{AB}}\sup_{\rho^{ABCR}}\bures\bigl(\cR_B^{AB}(\rho^{BCR}),~\rho^{ABCR}\bigr)= \delta_\ell(A)\label{eqn:infodist},
\end{align}
where the $\inf$ in Eqn.~(\ref{eqn:infodist}) is over all channels with support on $AB$.
\end{thm}
\noindent
The proof of Theorem~\ref{thm:infodist} in App.~\ref{app:A}
is quite similar to the ones in Refs.~\cite{kretschmann2008,BenyOreshkov2010}.

As a corollary of the information disturbance tradeoff,
we show that a region is locally correctable \emph{if and only if}
it can be disentangled from its complement by a unitary on the boundary.

\begin{corollary}\label{cor:disent}
Let $\Pi$ be a code space on $ABC=\Lambda$, and $R$ be a purifying auxiliary system.
Let $\cV_{B}^{B'B''} = V \cdot V^\dagger$ denote any isometry channel,
where $B'$ is some auxiliary system
and $B''$ is a copy of $AB$.
Then,
\begin{align}
\delta_\ell(A)= \inf_{\omega^{AB'}, \cV_B^{B'B''}} \sup_{\rho^{ABCR}}
\bures\bigl( \cV_B^{B'B''} (\rho^{ABCR}),~ \omega^{AB'} \otimes \rho^{B''CR} \bigr)
\label{eq:cor-decouplingUnitary}
\end{align}
where $\rho^{B''CR}$ is the same as $\rho^{ABCR}$
but supported on $B''$ instead of $AB$.
\end{corollary}
\noindent
The proof of the Corollary is in Appendix \ref{app:D}.

In Ref.~\cite{BPT}, it is shown that
for perfect error correcting codes defined by local commuting projectors,
there exists a disentangling unitary on the boundary of a correctable region
such that it maps any code vector into a product state.
Moreover the tensor factor of the product state
on the correctable region is the same for any code vector.
In view of our expression~\eqref{eq:cor-decouplingUnitary},
the result in Ref.~\cite{BPT} corresponds to a situation where $B''$ and $B'$
are subsystems of $B$.
Although our formulation requires $B'B''$ to have, in general, larger dimension than that of $B$,
our result immediately implies that of Ref.~\cite{BPT} when $\delta_\ell(A)= 0$.
This is because for local commuting projector codes,
any code state obeys an area law of entanglement Hartley (zeroth R\'enyi) entropy.
This means that there is a unitary transformation that can ``compress''
the Schmidt components in $B'B''$ of $\omega^{AB'} \otimes \rho^{B''CR}$ into $B$,
yielding a disentangling unitary within $B$.

The quantity $\delta_\ell(A)$ roughly expresses how weakly a correctable region $A$
is correlated with the far-separated region.
Indeed, it is the principle of error correction that the environment, interaction with which
would corrupt encoded information, learns nothing about the encoded states.
Put differently, the mutual information between a correctable region and the reference system
is always small.
The quantity $\delta_\ell(A)$ does not directly give the mutual information,
but is very similar (proof in Appendix~\ref{app:B}):
\begin{thm}[Decoupling-Correctibility]\label{thm:decoupling}
Let $A\subset\Lambda=ABC$ as in Fig.~\ref{fig:decoupling}, then
\begin{align}
\frac{1}{9} \delta_\ell(A)^2
\leq
\sup_{\rho^{ABCR}}\bures(\rho^{ACR},\rho^A\otimes\rho^{CR})
\leq
2 \delta_\ell(A)
\label{eqn:decoupling}
\end{align}
\end{thm}

Schumacher and Nielsen~\cite{SchumacherNielsen1996Quantum}
defined the coherent information,
which is equivalent to the mutual information $I(A:R)$
between the purifying system $R$ and the correctable region $A$.
They argued that if the coherent information does not degrade,
then perfect error correction is possible.
Schumacher and Westmoreland~\cite{Schumacher2002}
generalized this to an approximate setting,
and argued that the change in the coherent information bounds
the recovery fidelity.
However, in their argument
the recovery map was constructed from a code state,
and was not shown to be applicable for all code states.
This is not satisfactory for error correction
since any correcting map must not know about the code state.
In the perfect error correction setting,
this is no longer an issue due to the Knill-Laflamme
condition~\cite{KnillLaflamme1997Theory}.
(See also \cite{KitaevShenVyalyi2002CQC}.)
We were unable to find a reference that proves the existence of a recovery map
that would work for all code vectors in the approximate error correction setting,
so we proved our own in Theorem~\ref{thm:decoupling}.
Schumacher and Westmoreland's claim follows qualitatively
with a recovery map that depends only on $\Pi$
because~\cite{MuellerLennertDupuisSzehrFehrTomamichel2013, Beigi2013Sandwiched}
\begin{align}
I_\rho(A:R)
&\ge -2 \log \fidelity( \rho^{AR},~ \rho^A \otimes \rho^R ) \nonumber\\
&\ge 2 \bures^2( \rho^{AR},~ \rho^A \otimes \rho^R ).
\end{align}
Note that the continuity of the mutual information (see Appendix~\ref{sec:continuity-MI})
implies
\begin{align}
I_\rho(A:R)
\le & 18 \sqrt{2}\, \delta_\ell(A) \log\biggl(\frac{\dim \Pi }{ 2 \sqrt{2} \delta_\ell(A)} \biggr) \,.
\end{align}
No attempt has been made to optimize this inequality, and in fact we will only use a weaker statement with $\log \dim \Pi = k \log 2 \ge 1$ and $\delta_\ell(A) \le 1/e$,
\begin{align}
\label{eq:continuity-MI}
	I_\rho(A:R) \le O\bigl( k \delta_\ell(A) \log(1/\delta_\ell(A)) \bigr)\,
\end{align}
since some of the formulas become simpler.

\subsection{Cleanability as an alternative notion of correctability}

Here we characterize the correctability of a region
by logical operator avoidance.
If all logical operations can be done outside a region $A$,
then, tautologically, for any logical operator $U$
one can achieve the equivalent action on the code space
without touching $A$. Thus, the region $A$ is {\em cleaned} of $U$.
In a perfect error correction setting,
it is known that a correctable region can be cleaned of any logical operator,
and hence the name ``cleaning lemma''~\cite{BravyiTerhal2009, YoshidaChuang2010Framework, HaahPreskill2012tradeoff}.
We prove a generalization of the cleaning lemma in the approximate setting,
and complement it with the converse statement.
This establishes another error correction criterion
based on the support of logical operators.

\begin{defi}
An operator $U$ is {\bf logical}
if it commutes with the code space projector $\Pi$.
\end{defi}

\begin{thm}[A correctable region avoids logical operators.]\label{thm:cleaning1}
Suppose $A$ is a $(\delta,\ell)$-correctable region of the lattice $ABC=\Lambda$;
i.e., there exists $\cR^{AB}_B$ such that
\begin{align}
\sup_{\rho^{ABC}} \bures\bigl(\cR^{AB}_B( \rho^{BCR}),~\rho^{ABCR}\bigr) \leq \delta
\end{align}
where $AB$ is the $\ell$-neighborhood of $A$.
Then, for any logical unitary  $U^{ABC}$, the pull-back
$V^{BC}=(\cR_B^{AB})^*(U^{ABC})$ satisfies
\begin{align}
&\norm{(U^{ABC} - V^{BC})\Pi} \le 4 \sqrt \delta, \label{eq:diff-unres-codomain}\\
&\norm{\Pi (U^{ABC} - V^{BC})} \le 4 \sqrt \delta. \nonumber
\end{align}
\end{thm}

The converse without the locality of the recovery map is also true:
\begin{thm}[A region avoiding logical operators is correctable.]\label{thm:cleaning2}
Given the lattice $AB=\Lambda$, suppose for any logical unitary $U^{AB}$
there exists an operator $\norm{V^B} \leq 1$ supported on $B$ such that
\begin{align}
\norm{(U^{AB} - V^B) \Pi} \le \delta,\nonumber\\
\norm{\Pi (U^{AB} - V^B)} \le \delta.
\end{align}
Then there exists $\omega^A$ such that
\begin{align}
\sup_{\rho^{ABR}} \onenorm{\rho^{AR} - \omega^A \otimes \rho^R} \le 5 \delta,
\end{align}
and $A$ is $\sqrt{5\delta/2}$-correctable.
\end{thm}
\noindent
Theorems~\ref{thm:cleaning1} and \ref{thm:cleaning2} are proved in Appendix \ref{app:C}.

\section{Tradeoff bounds for locally correctable codes}\label{sec:tradeoff}

Our first main result is a tradeoff bound between the number of encoded qubits
and the distance of a code on the lattice.
The proof is similar in essence to the tradeoff bound proved for local commuting projector codes
by Bravyi, Poulin, and Terhal~\cite{BPT}.
However, the details vary in several respects
because we are dealing with approximate error correcting codes rather than exact ones.
Especially, Ref.~\cite{BPT} uses an algebraic decomposition of the Hilbert space 
resulting from the representation theory of commuting operators.
In the approximate setting we cannot use such a decomposition,
but instead we use a decoupling characterization of locally correctable codes.
We state and prove the theorem for a two-dimensional Euclidean lattice of qubits for clarity of presentation, but the proof generalizes easily to $D$-dimensional lattices with $D \ge 2$.
Having a Euclidean geometry will be important in the final statement
because we are going to use the fact that the surface area of a ball of radius $r$
grows like $r^{D-1}$.
We note that the technique below, which is borrowed from Ref.~\cite{BPT},
will be applicable for other geometries.

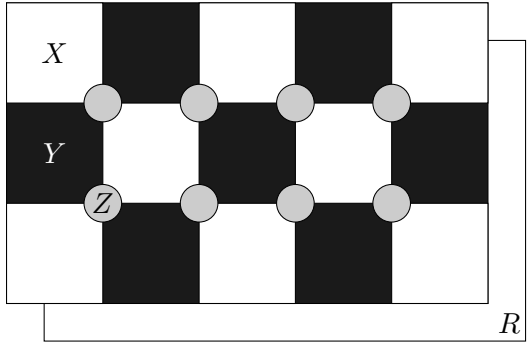
\begin{figure}
\centering
\begin{tikzpicture}
	[a/.style ={fill=white},
	b/.style ={fill=black!90!white},
	c/.style ={fill=black!20!white}]
	\def\xx{6.4/5} 
	\def\yy{4/3} 
	\draw[shift={(.5,-.5)}] (0,0) rectangle (6.4,4);
	\draw[b] (0,0) rectangle (6.4,4);
	\foreach \x in {0,2,4} {
	\foreach \y in {0,2}
		\draw[shift={(\x*\xx,\y*\yy)},a] (0,0) rectangle (\xx,\yy);
	}
	\draw[shift={(1*\xx,\yy)},a] (0,0) rectangle (\xx,\yy);
	\draw[shift={(3*\xx,\yy)},a] (0,0) rectangle (\xx,\yy);
	\foreach \x in {1,2,3,4} {
	\foreach \y in {1,2}
		\draw[shift={(\x*\xx,\y*\yy)},c] (0,0) circle (.25cm);
	}
	\node at (\xx/2,5*\yy/2) {$X$};
	\node[color=white] at (\xx/2,3*\yy/2) {$Y$};
	\node at (\xx,\yy) {$Z$};
	\node[below right] at (6.4,0) {$R$};
\end{tikzpicture}
    \caption{Decomposition of the lattice in 2D leading to the tradeoff bound of Theorem \ref{thm:ABPT}. }
    \label{fig:kd2n}
\end{figure}

\begin{thm}\label{thm:ABPT}
Let $\Pi$ be a local error correcting code with parameters $[[n,k,d,\delta,\ell]]$
defined on a $D$-dimensional Euclidean lattice of qubits.
It holds that
\begin{align}
\left(1 - c \frac{n\delta}{d} \log\frac{d}{n\delta}  \right) k d^{\frac{2}{D-1}} \leq  c' n \ell^{\frac{2D}{D-1}}
\label{eqn:BPT}.
\end{align}
where $c, c' > 0$ are constants independent of any parameters $n,k,d,\delta,\ell$.
\end{thm}
Obviously, this is meaningful only if the factor in the parenthesis is positive;
we must have sufficiently small $\delta$.
For local commuting projector codes,
$\delta=0$ for $\ell$ such that every local projector is contained in $\ell \times \ell$ square.
We then recover the bound $kd^{2/(D-1)} \le O(n)$ of Ref.~\cite{BPT}.
Note that we did not need to invoke any form of local generators of the code, rather we based our proof exclusively on the natural locality structure of the subspace.
If $\delta = \delta(\ell) \sim \exp(-\ell/\xi)$
(as we later show is the case for perturbations of commuting projector codes),
then we obtain
\begin{align}
kd^{\frac{2}{D-1}} \le O\Bigl(n (\log n)^{\frac{2D}{D-1}} \Bigr)
\end{align}
by choosing
$\ell \sim \xi \log n \sim \xi \log (L)$,
where $L = n^{1/D}$ is the linear system size.

The proof largely consists of three steps.
We divide the whole lattice into three subsystems $X,Y,Z$,
each of which is a collection of disjoint correctable simply connected regions.
They are depicted in Fig.~\ref{fig:kd2n} for $D=2$.
The first step is to choose the disconnected components of $X$ and $Y$ as large as possible.
The second step is to show that the union $X$ or $Y$ of those components is still correctable.
For this, each component should be sufficiently separated.
In both steps, the locality of the error correcting map will be important.
The third step is to employ a technique in a proof of the quantum Singleton bound~\cite{Preskill1999QEC}
to bound the code space dimension by the volume of $Z$.

We will need two lemmas below,
which allow us to construct regular correctable regions
with weight larger than the distance.
The \textit{expansion lemma} gives a prescription on how to grow a correctable region,
and the \textit{union lemma} shows that the union of two distant correctable regions is again correctable.

\begin{lem}[Expansion Lemma]\label{lem:expansion}
Let $AB = A^{+\ell}$ be the $\ell$-neighborhood of a region $A$.
If $A$ is $(\epsilon_A, \ell)$-correctable and $B$ is $(\epsilon_B,\ell)$-correctable,
then $A \cup B$ is $(\epsilon_A + \epsilon_B, \ell)$-correctable.
\end{lem}

\begin{proof}
Let  $ABC=A^{+2\ell}$, and let $D$ be the rest of the lattice.
Fix $\omega^A$ that saturates the optimum for $\delta_l(A)$.
From the correctability of region $B$, there exists a local recovery map $\cR_{AC}^{ABC}$ such that
\begin{align}
\bures\bigl( \rho^{ABCDR},~\cR_{AC}^{ABC}(\rho^{ACDR}) \bigr) \le \epsilon_B
\end{align}
for any code state $\rho^{ABCDR}$.
Define a channel $\cR_C^{ABC}$ by
\begin{align}
\cR_C^{ABC}\bigl(\sigma^C\bigr) := \cR_{AC}^{ABC}\bigl(\omega^A \otimes \sigma^C\bigr)
\end{align}
for any state $\sigma^C$.
Now, for an arbitrary code state $\rho^{ABCDR}$
\begin{align}
\bures\bigl( & \cR_C^{ABC} (\rho^{CDR}),~ \rho^{ABCDR} \bigr)\nonumber\\
=&
\bures\bigl( \cR_{AC}^{ABC}(\omega^A \otimes \rho^{CDR}),~\rho^{ABCDR} \bigr)\nonumber\\
\le&
\bures\bigl( \cR_{AC}^{ABC}(\omega^A \otimes \rho^{CDR}),~\cR_{AC}^{ABC}(\rho^{ACDR}) \bigr)\nonumber\\
&+
\bures\bigl( \cR_{AC}^{ABC}(\rho^{ACDR}),~\rho^{ABCDR} \bigr)\nonumber\\
\le&
\bures\bigl( \omega^A \otimes \rho^{CDR},~\rho^{ACDR} \bigr)
+
\epsilon_B\nonumber\\
\le&
\epsilon_A + \epsilon_B,
\end{align}
where we used monotonicity of the Bures distance.
Therefore, the map $\cR_{C}^{ABC}$ recovers $AB$ from $C$
up to error $\epsilon_A + \epsilon_B$.
\end{proof}

\begin{lem}[Union Lemma]\label{lem:union}
Suppose two regions $A$ and $B$ are separated by distance at least $\ell$,
where $A$ is $(\epsilon_A,\ell)$-correctable and $B$ is $\epsilon_B$-correctable.
Then the union $AB$ is $(\epsilon_A + \epsilon_B)$-correctable.
\end{lem}
\begin{proof}
Fix $\omega^A$ and $\omega^B$
that saturate the optima for $\delta_\ell(A)$ and $\delta_\ell(B)$.
The claim is proved as
\begin{align}
\bures( \rho^{ABR},~ \omega^A \omega^B \rho^R )
\le&
\bures( \rho^{ABR},~ \omega^A \rho^{BR} )\nonumber\\
&+ \bures( \omega^A \rho^{BR} ,~ \omega^A \omega^B \rho^R)\nonumber\\
\le&
\epsilon_A
+\bures( \rho^{BR} ,~ \omega^B \rho^R ) \\
\le&
\epsilon_A + \epsilon_B
\end{align}
for any code state $\rho^{ABCR}$.
\end{proof}

\begin{proof}[Proof of Theorem~\ref{thm:ABPT}.]
The statement becomes vacuous if $d \le \ell$.
So, assume $d > \ell$.

(Step 1)
Since $d > \ell$, a single site is $(\delta, \ell)$-correctable.
Its boundary contains $O(\ell^D)$ qubits.
If this is less than $d$, then it is also $(\delta,\ell)$-correctable,
and by Lemma~\ref{lem:expansion}
an enlarged hypercube of linear size of order $\ell$
is $(2\delta, \ell)$-correctable.
After $m$ steps of induction, we obtain a hypercube of linear size of order $m \ell$
that is $(O(m \delta), \ell)$-correctable.
The induction must stop if the hypercube is so large that the boundary ``area''
$O(m^{D-1} \ell^D)$ becomes larger than $d$.
Thus, $m$ can be as large as $O\left( \left(d/\ell^D\right)^{\frac{1}{D-1}} \right)$.
Therefore, we conclude that any hypercube of linear size
$O\left( (d/\ell)^{\frac{1}{D-1}} \right)$
is
$\left( O\left(\delta \left(d/\ell^D\right)^{\frac{1}{D-1}} \right), \ell \right)$-correctable.

\begin{figure}
\centering
\begin{tikzpicture}
	\def\rr{.38} 
	\draw[shift={(.5,-.5)}] (-3.2,-2) rectangle (3.2,2);
	\draw[fill=white] (-3.2,-2) rectangle (3.2,2);
	\foreach \r in {1,...,4} {
		\draw (-.1-\r*\rr,-.1-\r*\rr) rectangle (.1+\r*\rr,.1+\r*\rr);
	}
	\node at (0,0) {$A_0$};
	\node[anchor=south west] at (.2+\rr/2,.05+\rr/2) {${}_{B_1}$};
	\node[anchor=south west] at (.2+3*\rr/2,.05+2*\rr/2) {${}_{B_2}$};
	\node[anchor=south west] at (.2+5*\rr/2,.05+3*\rr/2) {${}_{B_3}$};
	\node[anchor=south east] at (3.2,-2) {$C$};
	\node[below right] at (3.2,-2) {$R$};
\end{tikzpicture}
\caption{Construction of the largest correctable square by successively adding rings of size at most $d$.}
    \label{fig:rings}
\end{figure}
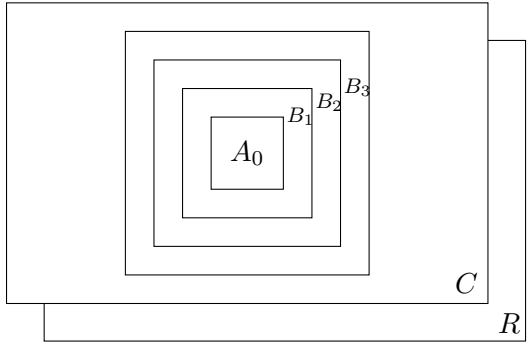

(Step 2) Consider a decomposition of the lattice into three regions $XYZ=\Lambda$,
where $X$ and $Y$ are, respectively,
the unions of disconnected hypercubes of linear size $O\left( (d/\ell)^{\frac{1}{D-1}} \right)$
constructed in Step~1 with their $(D-2)$-dimensional ``corners'' of width $\ell$ removed.
The hypercubes of $X$ and $Y$ are arranged in a high-dimensional checkerboard.
$Z$ is the rest of the lattice.
See Fig.~\ref{fig:kd2n} for an illustration of the decomposition for $D=2$.
Let $L$ be the linear system size.
The collections $X$ and $Y$ consist of
$O\left( L (\ell/d)^\frac{1}{D-1} \right)^D$ hypercubes, respectively.
Applying Lemma~\ref{lem:union} inductively over all squares of $X$,
we see that $X$ is $\epsilon$-correctable where
\begin{align}
\epsilon
= O(1) \left( L (\ell/d)^\frac{1}{D-1}\right)^D \cdot \left( \delta (d/\ell^D)^\frac{1}{D-1} \right)
= O(1) \frac{n \delta}{d} .
\label{eq:eps-of-union}
\end{align}
Similarly, $Y$ is also $\epsilon$-correctable.

(Step 3)
Let $\rho^{XYZ} = \Pi / \dim \Pi$ be the maximally mixed code state,
and $\rho^{XYZR}$ be a purification.
By definition, $S(\rho^R) = \log \dim \Pi = k \log 2$.
Since $X$ and $Y$ are $\epsilon$-correctable, respectively, by Eq.~\eqref{eq:continuity-MI},
the mutual information to $R$ must be small:
\begin{align}
S(\rho^X) + S(\rho^R) - S(\rho^{XR} ) &\le O(1) k \epsilon \log (1/\epsilon), \\
S(\rho^Y) + S(\rho^R) - S(\rho^{YR} ) &\le O(1) k \epsilon \log (1/\epsilon).
\end{align}
Adding the two inequalities and using the fact that $\rho^{XYZR}$ is pure,
we obtain
\begin{align}
S(\rho^X)+S(\rho^Y) + 2k \log 2 - S(\rho^{YZ}) -S(\rho^{XZ}) \le O(1) k \epsilon \log(1/\epsilon).\nonumber
\end{align}
By the subadditivity of entropy (where $\tilde{O}$ hides log factors), it follows that
\begin{align}
k(1- \tilde O(\epsilon)) \le O(1) S(\rho^Z) \le O( |Z| ).
\end{align}
The region $Z$ consists of a grid of ``bar segments'',
separating the individual hypercubes in $X$ and $Y$.
A bar segment has two sides of length $\ell$ and $D-2$ sides of length $O(d/\ell)^{1/(D-1)}$.
There are as many bar segments as there are the hypercubes of $X$ and $Y$,
the number of which is $O(n) (\ell/d)^{D/(D-1)}$.
Therefore,
\begin{align}
k( 1- \tilde O(\epsilon))
\le
O(n) (\ell/d)^\frac{D}{D-1} \cdot \ell^2 \cdot O(d/\ell)^\frac{D-2}{D-1}
=
O(n) \ell^{\frac{2D}{D-1}}d^{\frac{-2}{D-1}}.
\end{align}
This complete the proof of Theorem~\ref{thm:ABPT}.
\end{proof}

\subsection{Support of Logical Operators}

As corollaries of the preceding results,
we can also derive constraints on the support of logical operators.

Bravyi and Terhal~\cite{BravyiTerhal2009} have shown that
for stabilizer codes defined by local stabilizer generators on a Euclidean lattice
there always exists a nontrivial logical operator that is supported on
a thin slab ($(D-1)$-dimensional).
This result was generalized to local commuting projector codes~\cite{HaahPreskill2012tradeoff}.
In two-dimensional lattices,
this suggests that under generic interaction of the system with an environment
there would be a process subject to a constant energy penalty
that implements a nontrivial logical operation, i.e., an error, on the encoded state.
This is a strong argument against \emph{self-correction} in two-dimensional systems
under a thermalizing interaction
where the code space is the ground space.
In other words, either bit-flip errors or phase errors would occur by thermalization
on an encoded qubit at a rate independent of code distance.

However, this does not immediately rule out the possibility that
only one of bit or phase information is corrupted by thermalization,
but the other information is protected.
Indeed, any ferromagnetic (symmetry broken) system
is protected from bit-flip errors.
For this observation,
Haah and Preskill~\cite{HaahPreskill2012tradeoff} have asked whether
it is possible to have a \emph{partially self-correcting}
quantum memory in two-dimensions
where the code distance is high, but, e.g., bit-flip errors are suppressed under thermalization.
It was found that in any two-dimensional local commuting projectors code
if the code distance $d$, then \emph{all} logical operators
can be supported on a region that contains $\tilde d = O(L^2 / d)$ physical qubits.
This is a negative result towards partial self-correction
because having a large distance $d \sim L$ implies that any other logical operator
lives on a network of finitely many string-like regions.

We generalize these results using the machinery we have developed.

\begin{thm}
For any $(\delta,\ell)$-correctable code $\Pi$ with $\dim \Pi > 1$
on a $D$-dimensional lattice of linear size $L$,
if $10 L \delta < \ell$, then the code distance is bounded from above by $5\ell L^{D-1}$.
\label{thm:code-distance-bound}
\end{thm}
\begin{proof}
The argument is essentially one-dimensional,
and it suffices to prove the theorem for $D=1$.
Suppose on the contrary to the claim
that any segment of length $5\ell$ is $(\delta,\ell)$-correctable.
In particular, a segment of length $5\ell$ is $(\delta,\ell)$-correctable,
and its two-component boundary $B$ of size $2\ell$ is also $(\delta,\ell)$-correctable.
The expansion lemma~\ref{lem:expansion} implies that
the union $AB$ of length $7\ell$ is $(2\delta,\ell)$-correctable.
The two-component boundary of $AB$ is also $(\delta,\ell)$-correctable,
and we can again apply the expansion lemma~\ref{lem:expansion}
to have $(3\delta,\ell)$-correctable region of length $9\ell$.
After $O(L/\ell)$ times of iteration,
we see that entire system is $(\delta L / \ell)$-correctable.
This is a contradiction
since the entire system is certainly not $\epsilon$-correctable with $\epsilon < 1$.
\end{proof}

The conclusion is the most meaningful when applied to
a family of codes parametrized by $L$,
which in turn makes it necessary for the theorem
that $\delta$ be parametrically small in $\ell$.
For example, if $\delta \sim e^{-\ell}$,
then we can choose $\ell = \log L$
so that $L \delta \ll \ell$ and $d \le O(L^{D-1}\log L)$.

This code distance bound is intimately related to
the absence of topological order (without any symmetry)
in one dimension $D=1$. If we regard the ground space as a code space,
then having a small code distance means that
there is an operator of small support that takes
different expectation values for distinct ground states,
indicating that the degeneracy will be lifted 
upon perturbation by that operator.
In particular,
if the ground space is strictly locally correctable 
(i.e., $\delta=0$ for some constant $\ell$),
which is the case if the quantum phase of the one-dimensional system is represented by a commuting Hamiltonian
by the results in Section~\ref{sec:robustCommProjCodes} below,
then the ground space degeneracy, if any, is lifted by perturbations.
For fermionic systems, this only says the ground space is not locally correctable,
or a degeneracy-lifting local operator of fermion parity \emph{even or odd} should exist.

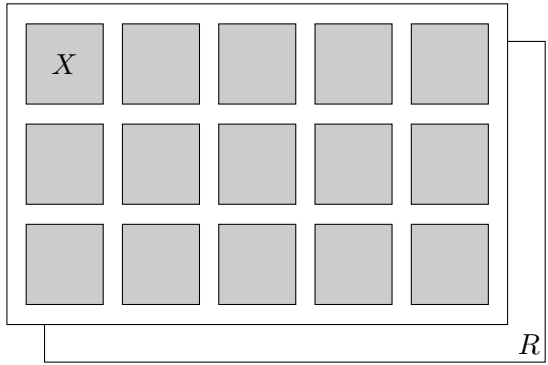
\begin{figure}
\centering
\begin{tikzpicture}
	[a/.style ={fill=white},
	b/.style ={fill=black!90!white},
	c/.style ={fill=black!20!white}]
	\def\xx{6.4/5} 
	\def\yy{4/3} 
	\draw[shift={(.5,-.5)}] (-.1*\xx,-.1*\yy) rectangle (6.4+.1*\xx,4+.1*\yy);
	\draw[a] (-.1*\xx,-.1*\yy) rectangle (6.4+.1*\xx,4+.1*\yy);
	\foreach \x in {0,...,4} {
	\foreach \y in {0,1,2}
		\draw[shift={(\x*\xx,\y*\yy)},c] (.1*\xx,+.1*\yy) rectangle (.9*\xx,.9*\yy);
	}
	\node at (\xx/2,5*\yy/2) {$X$};
	\node[below right] at (6.4+.1*\xx,-.1*\yy) {$R$};
\end{tikzpicture}
    \caption{Decomposition of the lattice used for Theorem~\ref{thm:logical-operator-tradeoff}.}
    \label{fig:HP}
\end{figure}

\begin{thm}
For any $(\delta,\ell)$-correctable code of code distance $d$
on a $D$-dimensional lattice with Euclidean geometry of linear size $L$,
there exists a region $Y$ that contains $\tilde d$ qubits such that
every logical operator $U$ can be approximated by an operator $V$ on $Y$
where
\begin{align}
& \norm{(U - V)\Pi} \le O\left(  \sqrt{n\delta/d} \right) \nonumber \\
&  \tilde d d^{\frac{1}{D-1}}\le O( n \ell^{\frac{D}{D-1}} )
\end{align}
\label{thm:logical-operator-tradeoff}
\end{thm}
\noindent
If $\delta$ becomes exactly zero at some $\ell$, independent of $L$,
then the conclusion becomes that of Ref.~\cite{HaahPreskill2012tradeoff}.
\begin{proof}
We divide the system as in Fig.~\ref{fig:HP}, where each hypercube is separated by a distance at least $\ell$.
Each hypercube in $X$ has linear size at most $O( d/\ell )^{\frac{1}{D-1}}$,
and is at least $\left(O(\delta(d/\ell^D)^{\frac{1}{D-1}} ), \ell \right)$-correctable
by Step~1 in the proof of Theorem~\ref{thm:ABPT}.
The union $X$ of all such squares is $O( n \delta / d )$-correctable
by the union lemma~\ref{lem:union}.
(See Eq.~\eqref{eq:eps-of-union}).
Then, Theorem~\ref{thm:cleaning1} (cleaning lemma) implies that
the complement, on which there are
$O(n) (\ell/d)^\frac{D}{D-1} \cdot \ell \cdot O(d/\ell)$ qubits,
supports all logical operators to accuracy $O(\sqrt{n\delta/d})$ in operator norm.
\end{proof}

\section{Flexible strings imply finite degeneracy}\label{sec:flex}

Here we apply Theorem~\ref{thm:cleaning2}
to topologically ordered two-dimensional systems,
to show that under a mild condition
the degeneracy can be at most a constant independent of system size.
This constant degeneracy holds
for all examples we know of with a stable ground state subspace
subject to arbitrary perturbations,
and is intimately related to one of the core assumptions of
algebraic anyon theories (modular tensor categories)
that there are only finitely many superselection sectors (simple objects).
Without the finiteness, any interesting computation in the algebraic anyon theory
would contain an infinite sum,
rendering the theory substantially different from what we
understand with the finiteness.
For instance Vafa's theorem~\cite{vafa} stating that the topological spin
is rational would not hold without the finiteness.

We consider periodic boundary conditions,
though it will not be too important
whether the boundary conditions are open or periodic.
The overall topology of the system is a 2-torus.
In a topologically ordered two-dimensional system,
it is expected that
\begin{itemize}
\item[$(\circ)$] there exists a complete set of operators acting within the ground space $\Pi$
(logical operators)
such that they are supported along a narrow strip (string)
wrapping around non-contractible loops of the system.
\end{itemize}
It is also expected that the string can be bent while implementing the same action on the ground space
as long as the support remains isotopic to the initial one.
These properties can be regarded
as the mathematical defining properties of topological order.

\begin{figure}
\centering
\begin{tikzpicture}
	[a/.style ={fill=white},
	b/.style ={fill=black!90!white},
	c/.style ={fill=black!20!white}]
	\def\xx{1.5} 
	\def\yy{1.5} 
	\draw[shift={(.5,-.5)}] (0,0) rectangle (2*\xx, 2*\yy);
	\draw[b] (0,0) rectangle (2*\xx, 2*\yy);
	\draw[shift={(0*\xx,1*\yy)}, a] (0,0) rectangle (\xx,\yy);
	\draw[shift={(1*\xx,0*\yy)}, a] (0,0) rectangle (\xx,\yy);
	\draw[shift={(1*\xx,1*\yy)},c] (0,0) circle (.25*\xx);
	\draw[c] (0,0) -- (.25*\xx,0) arc [start angle=0, end angle=90, radius=.25*\xx] -- (0,0);
	\draw[c] (2*\xx,0) -- (1.75*\xx,0) arc [start angle=180, end angle=90, radius=.25*\xx] -- (2*\xx,0);
	\draw[c] (0,2*\yy) -- (.25*\xx,2*\yy) arc [start angle=0, end angle=-90, radius=.25*\xx] -- (0,2*\yy);
	\draw[c] (2*\xx,2*\yy) -- (1.75*\xx,2*\yy) arc [start angle=180, end angle=270, radius=.25*\xx] -- (2*\xx,2*\yy);
	\draw[c] (.75*\xx,0) arc [start angle=180, end angle=0, radius=.25*\xx] -- (.75*\xx,0);
	\draw[c] (0,1.25*\yy) arc [start angle=90, end angle=-90, radius=.25*\xx] -- (0,.75*\yy);	\draw[c] (2*\xx,.75*\yy) arc [start angle=270, end angle=90, radius=.25*\xx] -- (2*\xx,1.25*\yy);
	\draw[c] (1.25*\xx,2*\yy) arc [start angle=0, end angle=-180, radius=.25*\xx] -- (.75*\xx,2*\yy);
	\node at (\xx/2, 1.5*\yy) {$X$};
	\node[color=white] at (0.5*\xx, 0.5*\yy) {$Y$};
	\node at (\xx,\yy) {$Z$};
	\node[below right] at (2*\xx,0) {$R$};
\end{tikzpicture}
\caption{
Decomposition of the lattice in 2D used to define flexible logical operators.
The diameter of the disk $Z$ is $\ell$.
}
\label{fig:four-square}
\end{figure}
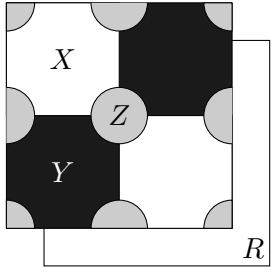

Assuming the existence of a complete set of such operators acting within $\Pi$,
we now bound the degeneracy.
Consider four squares with corners removed that fills the system as in Fig.~\ref{fig:four-square}.
The corners are occupied by four regions comprising $Z$.
The two squares on one diagonal comprise $X$,
and the other two squares comprise $Y$.
The size of a disk in $Z$ has diameter $\ell$ so that the two squares in either $X$ or $Y$
are separated by distance $\ell$.
According to the property $(\circ)$,
we can find a complete set of logical operators supported on $XZ$
since $XZ$ contains two essential non-contractible loops of the torus.
It may be unreasonable to assume that all operators that strictly preserve the ground space
are supported entirely on $XZ$.
Instead, we may reasonably assume that such an operator can be approximated by
one that is supported entirely on $XZ$, and the approximation
becomes better as we increase $\ell$.

Hence, we are well motivated to define the notion of flexible logical operators as follows:
\begin{defi}
A subspace $\Pi$ on a two-dimensional system {\bf admits flexible (logical) operators}
if
for any logical unitary operator $U^{XYZ}$
there exist operators $V_1^{YZ}$ supported on $YZ$ and
$V_2^{XZ}$ on $XZ$
such that
$\norm{V_i} \le 1$,
$\norm{\Pi(U-V_i)} \le \epsilon_\ell$, and
$\norm{(U-V_i)\Pi} \le \epsilon_\ell$,
where $i = 1,2$ and
$\epsilon_\ell$ is independent of system size and vanishes as $\ell \to \infty$.
\label{def:flex}
\end{defi}

From this definition Theorem~\ref{thm:cleaning2} asserts that the complement $Y = (XZ)^c$
is $\sqrt{5\epsilon_\ell / 2}$-correctable.
Interchanging $X$ and $Y$, we conclude that $Y$ is also $\sqrt{5\epsilon_\ell /2}$-correctable.
Next, we employ the technique of Step 3 in the proof of Theorem~\ref{thm:ABPT}.
Recalling the continuity of mutual information Eq.~\eqref{eq:continuity-MI},
we deduce that
\begin{align}
\left(1- \tilde O\bigl(\epsilon_\ell^{1/2}\bigr) \right) \log \dim \Pi \le S(\rho^{Z}) \le \ell^2 \log 2 .
\end{align}
Choosing $\ell$ sufficiently large (and hence $\epsilon_\ell$ sufficiently small), 
we have
\begin{align}
\dim \Pi \le \exp( O( \ell^2 ) ).
\end{align}

In summary, the bound on the degeneracy $\dim \Pi$
is determined by the width of the strip
that can support sufficiently faithful logical operators.
It is noteworthy that the error from restricting the logical operator
onto the strip only has to be suppressed to an absolute constant.
Also, the exponent $\ell^2$ being quadratic in $\ell$ is optimal;
$\ell^2$ copies of the toric code can be laid on a lattice
to achieve this quadratic scaling.

We remark that our argument that the assumption of flexible string operators
implies constant degeneracy is unique to two dimensions in the following sense.
A higher dimensional analogue of the flexible string operators
is a set of faithful logical operators supported on noncontractible hypersurfaces.
We can say logical operators are ``flexible'' if there are equivalent operators
on isotopic hypersurfaces.
This applies to every known example of a gapped phase
with robust ground state degeneracy.
Discrete gauge theories such as the three-dimensional toric code
have surface operators and line operators, 
all of which are supported
on the union of three deformable planes.
In general, however,
flexible logical operators on hypersurfaces 
do not imply constant degeneracy.
As a trivial example one can consider a stack of two-dimensional toric codes,
of which the string operators are lurking within hyperplanes.
The degeneracy (the code space dimension) is exponential in the height of the stack.
More complicated examples are quantum glass or fracton models~\cite{Chamon2005Quantum,
BravyiLeemhuisTerhal2011Topological,Haah2011Local,VijayHaahFu2015New},
for which the fact that there is a set of faithful logical operators on planes
can be shown by the cleaning lemma (Theorem~\ref{thm:cleaning1}) applied to the bulk.
Here, the degeneracy is exponential in the linear system size, too~\cite{Haah2012PauliModule}.

\section{Robustness of local commuting projector codes}
\label{sec:robustCommProjCodes}

The results so far use properties of subspaces,
and do not hinge on particular code constructions.
Here we will 
show that our results apply to a class of local approximate error correction codes 
that are in the same gapped phase as certain exact codes.
Let $H_0 = - \sum_{k\in\Lambda}h_k$ be a local unfrustrated commuting projector Hamiltonian. 
That is, each term is a projector $h_k = h_k^2$, commutes with any other term $[h_k, h_j]=0$,
and is supported on a region around site $k$ of diameter less than $w$ (the interaction range),
and any ground state $\ket \psi$ satisfies $h_k \ket{\psi} = \ket \psi$.
The code subspace is identified with the ground state subspace of $H_0$: $\Pi = \prod_{k \in \Lambda} h_k$.
We summarize the error correction properties of
these \emph{local commuting projector} codes in the following lemma:

\begin{lem} \label{lem:equivcond}
Let $\Pi$ be a commuting projector code,
and $ABC=\Lambda$ be decomposition of the lattice such that 
the distance between $A$ and $C$ is at least $\ell \ge w$, the interaction range
(\emph{e.g.}~as in Fig.~\ref{fig:rings}.)
Then the following are equivalent:
\begin{enumerate}
\item[$(i)$] Topological Quantum Order (TQO):
for any observable $O_A$ with support on $A$, 
any two ground states $\ket{\phi}$ and $\ket{\psi}$ give the same expectation value,
$\bra{\phi}O^A\ket{\phi} = \bra{\psi}O^A\ket{\psi}$.

\item[$(ii)$] Decoupling:
For any $\rho\in\Pi$, we have $I_\rho(A:CR)=0$.

\item[$(iii)$] Error correction:
 There exists a recovery map acting on $AB$
 such that $\cR_B^{AB}(\rho^{BC})=\rho^{ABC}$ for any $\rho\in\Pi$.

\item[$(iv)$] Disentangling unitary:
There exists a unitary $U^{B}$ and a state $\omega^{AB_1}$ such that for any $\rho\in\Pi$
we have $U^{B}\rho U^{B\dagger}=\omega^{AB_1}\otimes\rho^{B_2C}$.

\item[$(v)$] Cleaning: For any unitary $U$ preserving the code space,
there exists a unitary $V^{BC}$ such that $U|_\Pi = V^{BC}|_\Pi$.
\end{enumerate}
\end{lem}
The condition $(i)$ is Knill-Laflamme~\cite{KnillLaflamme1997Theory} condition for correctability.
The implication from $(iii)$ to $(iv)$ was derived in Ref.~\cite{BPT} with slightly larger $B$.
The implication from $(iv)$ to $(v)$ was derived in Ref.~\cite{HaahPreskill2012tradeoff} with a slightly larger $B$.
\begin{proof}
$(i)\Rightarrow(ii)$: $(i)$ is equivalent to $\Pi O^A \Pi = c(O^A) \Pi$ for some scalar $c(O^A)$.
Then, it suffices to look at the correlation function between $A$ and $C$;
\begin{align}
&\bra \psi O^A O^C \ket \psi \\
&= \Bra \psi \Pi O^A \left(\prod_{k :\ell\text{-away from }A} h_k\right) \left(\prod_{k :\ell\text{-away from }C} h_k\right) O^C \Pi \Ket \psi\\
&= \bra \psi \Pi O^A \Pi O^C \Pi \ket \psi \\
&= c(O^A) c(O^C) = \bra \psi O^A \ket \psi \bra \psi O^C \ket \psi .
\end{align}

$(ii) \Leftrightarrow (iii) \Leftrightarrow (iv)$ by Theorem~\ref{thm:infodist}, Corollary~\ref{cor:disent}, and Theorem~\ref{thm:decoupling}
with $\delta_\ell(A) = 0$.

$(iii) \Rightarrow (v)$ by Theorem~\ref{thm:cleaning1} with $\delta = 0$.

$(v) \Rightarrow (i)$: Theorem~\ref{thm:cleaning2} with $\delta=0$ implies that the reduced density matrix for $A$ is the same for all code states.
\end{proof}

We now consider sufficiently weak perturbations of commuting projector codes.
We assume that the code distance (Definition~\ref{def:correctability})
grows with the system size by a power law ($d = \Omega(n^\gamma)$ for some $\gamma > 0$).
In addition,
we assume that the Hamiltonian $H_0$ obeys the ``local TQO'' condition of Refs.~\cite{spirosbravyi2010,thespiros}.
These conditions ensure that a perturbed Hamiltonian $H_1 = H_0 + \sum_k V_k$
has a \emph{gapped energy spectrum} above the ground state subspace,
where the $V_k$ are local and $\norm{V_k} \le \epsilon$
for a constant $\epsilon$ that only depends on the spatial dimension
and the interaction range $w$, but the $V_k$ are otherwise arbitrary.

The local TQO condition is similar to the TQO condition,
but eliminates effects from the correctable region's boundary.
$H_0$ is said to obey local TQO if
\begin{align}
	&\Pi^{A^w} O^A \Pi^{A^w} = c(O^A) \Pi^{A^w}\\
	& \text{ where } \Pi^{A^w} = \prod_{k: \text{dist}(k,A) \le w} h_k
\end{align}
for any region $A$ of size less than $d$.
In other words, if any state $\ket \psi$ satisfies $h_k \ket \psi = \ket \psi$
for $h_k$ around a correctable region $A$, then the reduced density matrix for $A$ of $\ket \psi$
is determined and unique.

The gap stability result implies that
for any sufficiently weak perturbation $V$ 
there is a gap (independent of system size) 
in the energy spectrum of $H_1 = H_0+V$ above the $m$ lowest energy eigenstates
where $m = \dim \Pi_0$.
Furthermore, if $\Pi_1$ denotes the projector onto these $m$ lowest energy eigenstates of $H_1$,
then there is
a locality-preserving unitary $U$~\cite{hastings2004,hastings2005,bachmann2012}
such that
\begin{align}
 & \Pi_1 = U \Pi_0 U^\dagger,\\
 &  \norm{ U O^X U^\dag - V^{X^{+r}} O^X (V^{X^{+r}})^\dag } \leq c_1 \norm{O^X} \exp\left(- c_2 r / \log^2 r \right).
 \label{eqn:TNlocalU}
\end{align}
Here, $X$ is an arbitrary region,
and $V^{X^{+r}}$ is strictly supported on the $r$-neighborhood of $X$,
and $c_1, c_2 > 0$ are constants.
The particular form of the function on the right-hand side of \eqref{eqn:TNlocalU}
is from Ref.~\cite[Theorem~3.4]{bachmann2012}.

We are now in position to state the main theorem of this section:
\begin{thm}
Let $H$ be a local commuting frustration-free Hamiltonian,
whose ground state subspace $\Pi_0$ is an $[[n,k,d]]$ quantum error correcting code
with $d = \Omega(n ^\gamma)$ for some $\gamma >0$.
An arbitrary but sufficiently weak local pertubation of $H$
defines an $[[n,k,d-2\ell,\delta(\ell),\ell]]$
approximate error correcting code $\Pi_1$,
where $\delta(\ell)\leq c_1 e^{- c_2 \ell/\log^2 \ell}$ for some constant $c_1, c_2 > 0$.
\label{thm:perturbCP}
\end{thm}

\begin{proof}
The locality-preserving property of \eqref{eqn:TNlocalU}
can be cast into a more convenient form in the following.
For any channel $\mathcal X^{M}_M$
on a region $M$
there exists a channel $\mathcal Y^{M^{+r}}_{M^{+r}}$
such that
\begin{align}
\sup_{\eta^{ABCR}} \bures\left(
\mathcal U^\dagger \mathcal X^{M} \mathcal U (\eta^{ABCR}),~
 \mathcal Y^{M^{+r}}( \eta^{ABCR} )
 \right) \le \epsilon ,\label{eqn:localU}
\end{align}
where $\mathcal U$ denotes the unitary conjugation channel by $U$,
and $\eta$ is any state, not necessarily a code state.
This claim \eqref{eqn:localU} with $\epsilon = \sqrt{2 C e^{-cr / \log^2 r}}$
follows by setting $\mathcal Y = \mathcal V^{M^{+r}}$
from \eqref{eqn:TNlocalU} and the triangle inequality of trace norm.
Then, we can turn any correcting map for the unperturbed code space
to a correcting map for the perturbed one, by sacrificing the locality and accuracy a little:
\begin{lem}\label{lem:localpert}
If a region $A$ is $(\delta,\ell)$-correctable with respect to a code $\Pi$,
then the region $Z = A^{-r}$ is $(\delta + 2\epsilon,\ell + 2r)$-correctable
with respect to the code $U \Pi U^\dagger$.
\end{lem}
\noindent
The lemma is proved in Appendix~\ref{app:E}.
\end{proof}

The locality-preserving unitary $U$ in the sense of Eq.~\eqref{eqn:TNlocalU}
exists under a sole condition that there exists a gapped Hamiltonian path
whose ground spaces are $\Pi_0$ and $\Pi_1$~\cite{hastings2005}.
Hence, by Lemma~\ref{lem:localpert},
the notion of locally correctable region is an invariant property of gapped phases of matter.

\section{Outlook}

In this paper, we have introduced a coherent framework
for analyzing local approximate quantum error correcting codes on lattices,
which naturally includes weak local perturbations of commuting projector codes.
Based only on very general properties of recovery maps on the code subspace,
we have proved a tradeoff bound between $n$, $k$, $d$ (Theorem~\ref{thm:ABPT}),
and given a constraints on the shape of logical operators (Theorem~\ref{thm:logical-operator-tradeoff}),
resembling the results for commuting projector codes in Ref.~\cite{BPT,HaahPreskill2012tradeoff}.
Furthermore, we have shown that if the logical operators of a code in 2D 
can be approximated by flexible string operators of constant width,
then the degeneracy of the code subspace is constant,
supporting one of the core assumptions of algebraic anyon theories.

We have assumed that there is one qubit ($\mathbb C^2$) per unit volume of the lattice.
One can translate all our results to a situation where the local Hilbert space dimension is
$q \ge 2$, simply by redefining the unit length as $1 \to (\log q)^{1/D}$.
It is thus important for the last result on the degeneracy bound under the assumption of flexible strings
that the local Hilbert space dimension is finite.

A number of important problems remain open, and we outline a few below.

Eastin and Knill~\cite{eastinknill} have proved that
no exact error correcting code can have a universal set of transversal gates. 
This no-go result has been sharpened for local stabilizer codes on lattices by
Bravyi and K\"onig~\cite{BravyiKonig2013},
where it was shown that
any transversal logical gate in a $D$-dimensional local stabilizer code
is contained in the $D$th level of the Clifford hierarchy.
In general, it is important to understand
whether approximate error correcting codes would
provide an avenue to local universal fault-tolerance in 2D or 3D that circumvents these no-go results.

As remarked in the Introduction, our notion of local recovery map is in part motivated
by topological systems with anyons where pair-annihilation amounts to error recovery.
We indeed have shown that the local recovery map exists for the ground space of a local commuting projector Hamiltonian
and its perturbation.
However, we have not explicitly discussed any larger class of phases of mattter
that do not necessarily have commuting Hamiltonian representatives,
yet have a ground space that is a local approximate quantum error correcting code.
One may wonder if, for any gapped Hamiltonian with stable gap against perturbations,
the ground space can be regarded as a local approximate error correcting code.
A more concrete problem would be to ask
about the optimal relation between the parameters $\ell$ and $\delta$
under the local TQO condition for unfrustrated Hamiltonians~\cite{thespiros}.

Since our code is identified with a subspace of a physical Hilbert space,
we have excluded subsystem codes~\cite{Poulin2005Stabilizer}
or more general subalgebra codes~\cite{Beny2009}.
Also neglected is the thermal encoding,
which has been studied in regards to self-correcting quantum memory~\cite{dennis, alicki, brown}.
B\'eny and Oreshkov~\cite{Beny2009,BenyOreshkov2010}
have discussed information-disturbance tradeoffs for subalgebra codes,
but it remains to be seen how these would apply with the local recovery maps.
Extensions to these more general error correction schemes with approximate recovery maps
might be relevant in connection to field theories with holographic duals~\cite{HAPPY,Harlow2016}.

\subsection*{Acknowledgments}
STF was supported by the Australian Research Council via EQuS project number CE11001013
and by an Australian Research Council Future Fellowship FT130101744.
JH was supported by the Pappalardo Fellowship in Physics while at MIT.
MJK was supported by the Carlsberg fund and the Villum foundation.
IK's research at Perimeter Institute was supported by the Government of Canada through Industry Canada and by the Province of Ontario through the Ministry of Economic Development and Innovation.

\bibliographystyle{apsrev4-1}
\nocite{apsrev41Control}
\bibliography{aqec-tradeoff-ref,aqec-tradeoff-ref-revtex-extra}

\begin{thebibliography}{56}%
\makeatletter
\providecommand \@ifxundefined [1]{%
 \@ifx{#1\undefined}
}%
\providecommand \@ifnum [1]{%
 \ifnum #1\expandafter \@firstoftwo
 \else \expandafter \@secondoftwo
 \fi
}%
\providecommand \@ifx [1]{%
 \ifx #1\expandafter \@firstoftwo
 \else \expandafter \@secondoftwo
 \fi
}%
\providecommand \natexlab [1]{#1}%
\providecommand \enquote  [1]{``#1''}%
\providecommand \bibnamefont  [1]{#1}%
\providecommand \bibfnamefont [1]{#1}%
\providecommand \citenamefont [1]{#1}%
\providecommand \href@noop [0]{\@secondoftwo}%
\providecommand \href [0]{\begingroup \@sanitize@url \@href}%
\providecommand \@href[1]{\@@startlink{#1}\@@href}%
\providecommand \@@href[1]{\endgroup#1\@@endlink}%
\providecommand \@sanitize@url [0]{\catcode `\\12\catcode `\$12\catcode
  `\&12\catcode `\#12\catcode `\^12\catcode `\_12\catcode `\%12\relax}%
\providecommand \@@startlink[1]{}%
\providecommand \@@endlink[0]{}%
\providecommand \url  [0]{\begingroup\@sanitize@url \@url }%
\providecommand \@url [1]{\endgroup\@href {#1}{\urlprefix }}%
\providecommand \urlprefix  [0]{URL }%
\providecommand \Eprint [0]{\href }%
\providecommand \doibase [0]{http://dx.doi.org/}%
\providecommand \selectlanguage [0]{\@gobble}%
\providecommand \bibinfo  [0]{\@secondoftwo}%
\providecommand \bibfield  [0]{\@secondoftwo}%
\providecommand \translation [1]{[#1]}%
\providecommand \BibitemOpen [0]{}%
\providecommand \bibitemStop [0]{}%
\providecommand \bibitemNoStop [0]{.\EOS\space}%
\providecommand \EOS [0]{\spacefactor3000\relax}%
\providecommand \BibitemShut  [1]{\csname bibitem#1\endcsname}%
\let\auto@bib@innerbib\@empty
\bibitem [{\citenamefont {Terhal}(2015)}]{Terhal2015}%
  \BibitemOpen
  \bibfield  {author} {\bibinfo {author} {\bibfnamefont {B.~M.}\ \bibnamefont
  {Terhal}},\ }\bibfield  {title} {\enquote {\bibinfo {title} {Quantum error
  correction for quantum memories},}\ }\href {\doibase
  10.1103/RevModPhys.87.307} {\bibfield  {journal} {\bibinfo  {journal} {Rev.
  Mod. Phys.}\ }\textbf {\bibinfo {volume} {87}},\ \bibinfo {pages} {307}
  (\bibinfo {year} {2015})},\ \Eprint {http://arxiv.org/abs/1302.3428}
  {arXiv:1302.3428} \BibitemShut {NoStop}%
\bibitem [{\citenamefont {Gottesman}(2010)}]{gottesman2009}%
  \BibitemOpen
  \bibfield  {author} {\bibinfo {author} {\bibfnamefont {D.}~\bibnamefont
  {Gottesman}},\ }\bibfield  {title} {\enquote {\bibinfo {title} {An
  introduction to quantum error correction and fault-tolerant quantum
  computation},}\ }in\ \href@noop {} {\emph {\bibinfo {booktitle} {Quantum
  Information Science and Its Contributions to Mathematics}}},\ Vol.~\bibinfo
  {volume} {68},\ \bibinfo {editor} {edited by\ \bibinfo {editor}
  {\bibfnamefont {S.~J.}\ \bibnamefont {{Lomonaco, Jr.}}}}\ (\bibinfo
  {publisher} {American Mathematical Society},\ \bibinfo {year} {2010})\ pp.\
  \bibinfo {pages} {24--69},\ \Eprint {http://arxiv.org/abs/0904.2557}
  {arXiv:0904.2557} \BibitemShut {NoStop}%
\bibitem [{\citenamefont {Leung}\ \emph {et~al.}(1997)\citenamefont {Leung},
  \citenamefont {Nielsen}, \citenamefont {Chuang},\ and\ \citenamefont
  {Yamamoto}}]{leung1997}%
  \BibitemOpen
  \bibfield  {author} {\bibinfo {author} {\bibfnamefont {D.~W.}\ \bibnamefont
  {Leung}}, \bibinfo {author} {\bibfnamefont {M.~A.}\ \bibnamefont {Nielsen}},
  \bibinfo {author} {\bibfnamefont {I.~L.}\ \bibnamefont {Chuang}}, \ and\
  \bibinfo {author} {\bibfnamefont {Y.}~\bibnamefont {Yamamoto}},\ }\bibfield
  {title} {\enquote {\bibinfo {title} {Approximate quantum error correction can
  lead to better codes},}\ }\href {\doibase 10.1103/PhysRevA.56.2567}
  {\bibfield  {journal} {\bibinfo  {journal} {Phys. Rev. A}\ }\textbf {\bibinfo
  {volume} {56}},\ \bibinfo {pages} {2567--2573} (\bibinfo {year}
  {1997})}\BibitemShut {NoStop}%
\bibitem [{\citenamefont {Cr{\'e}peau}\ \emph {et~al.}(2005)\citenamefont
  {Cr{\'e}peau}, \citenamefont {Gottesman},\ and\ \citenamefont
  {Smith}}]{crepeau2005}%
  \BibitemOpen
  \bibfield  {author} {\bibinfo {author} {\bibfnamefont {C.}~\bibnamefont
  {Cr{\'e}peau}}, \bibinfo {author} {\bibfnamefont {D.}~\bibnamefont
  {Gottesman}}, \ and\ \bibinfo {author} {\bibfnamefont {A.}~\bibnamefont
  {Smith}},\ }\enquote {\bibinfo {title} {Approximate quantum error-correcting
  codes and secret sharing schemes},}\ in\ \href {\doibase 10.1007/11426639_17}
  {\emph {\bibinfo {booktitle} {Advances in Cryptology -- EUROCRYPT 2005: 24th
  Annual International Conference on the Theory and Applications of
  Cryptographic Techniques, Aarhus, Denmark, May 22-26, 2005. Proceedings}}},\
  \bibinfo {editor} {edited by\ \bibinfo {editor} {\bibfnamefont
  {R.}~\bibnamefont {Cramer}}}\ (\bibinfo  {publisher} {Springer Berlin
  Heidelberg},\ \bibinfo {address} {Berlin, Heidelberg},\ \bibinfo {year}
  {2005})\ pp.\ \bibinfo {pages} {285--301},\ \Eprint
  {http://arxiv.org/abs/quant-ph/0503139} {quant-ph/0503139} \BibitemShut
  {NoStop}%
\bibitem [{\citenamefont {Moore}\ and\ \citenamefont {Read}(1991)}]{moore1991}%
  \BibitemOpen
  \bibfield  {author} {\bibinfo {author} {\bibfnamefont {G.}~\bibnamefont
  {Moore}}\ and\ \bibinfo {author} {\bibfnamefont {N.}~\bibnamefont {Read}},\
  }\bibfield  {title} {\enquote {\bibinfo {title} {Nonabelions in the
  fractional quantum hall effect},}\ }\href {\doibase
  10.1016/0550-3213(91)90407-O} {\bibfield  {journal} {\bibinfo  {journal}
  {Nuclear Physics B}\ }\textbf {\bibinfo {volume} {360}},\ \bibinfo {pages}
  {362--396} (\bibinfo {year} {1991})}\BibitemShut {NoStop}%
\bibitem [{\citenamefont {Kitaev}(2006)}]{kitaev2006}%
  \BibitemOpen
  \bibfield  {author} {\bibinfo {author} {\bibfnamefont {A.}~\bibnamefont
  {Kitaev}},\ }\bibfield  {title} {\enquote {\bibinfo {title} {Anyons in an
  exactly solved model and beyond},}\ }\href {\doibase
  10.1016/j.aop.2005.10.005} {\bibfield  {journal} {\bibinfo  {journal} {Annals
  of Physics}\ }\textbf {\bibinfo {volume} {321}},\ \bibinfo {pages} {2--111}
  (\bibinfo {year} {2006})},\ \Eprint {http://arxiv.org/abs/cond-mat/0506438}
  {cond-mat/0506438} \BibitemShut {NoStop}%
\bibitem [{\citenamefont {Michalakis}\ and\ \citenamefont
  {Zwolak}(2013)}]{thespiros}%
  \BibitemOpen
  \bibfield  {author} {\bibinfo {author} {\bibfnamefont {S.}~\bibnamefont
  {Michalakis}}\ and\ \bibinfo {author} {\bibfnamefont {J.~P.}\ \bibnamefont
  {Zwolak}},\ }\bibfield  {title} {\enquote {\bibinfo {title} {Stability of
  frustration-free {H}amiltonians},}\ }\href {\doibase
  10.1007/s00220-013-1762-6} {\bibfield  {journal} {\bibinfo  {journal}
  {Communications in Mathematical Physics}\ }\textbf {\bibinfo {volume}
  {322}},\ \bibinfo {pages} {277--302} (\bibinfo {year} {2013})},\ \Eprint
  {http://arxiv.org/abs/1109.1588} {arXiv:1109.1588} \BibitemShut {NoStop}%
\bibitem [{\citenamefont {Hastings}\ and\ \citenamefont
  {Wen}(2005)}]{hastings2005}%
  \BibitemOpen
  \bibfield  {author} {\bibinfo {author} {\bibfnamefont {M.~B.}\ \bibnamefont
  {Hastings}}\ and\ \bibinfo {author} {\bibfnamefont {X.-G.}\ \bibnamefont
  {Wen}},\ }\bibfield  {title} {\enquote {\bibinfo {title} {Quasiadiabatic
  continuation of quantum states: The stability of topological ground-state
  degeneracy and emergent gauge invariance},}\ }\href {\doibase
  10.1103/physrevb.72.045141} {\bibfield  {journal} {\bibinfo  {journal}
  {Physical Review B}\ }\textbf {\bibinfo {volume} {72}},\ \bibinfo {pages}
  {045141} (\bibinfo {year} {2005})},\ \Eprint
  {http://arxiv.org/abs/cond-mat/0503554} {cond-mat/0503554} \BibitemShut
  {NoStop}%
\bibitem [{\citenamefont {Bravyi}\ \emph
  {et~al.}(2010{\natexlab{a}})\citenamefont {Bravyi}, \citenamefont
  {Hastings},\ and\ \citenamefont {Michalakis}}]{spirosbravyi2010}%
  \BibitemOpen
  \bibfield  {author} {\bibinfo {author} {\bibfnamefont {S.}~\bibnamefont
  {Bravyi}}, \bibinfo {author} {\bibfnamefont {M.~B.}\ \bibnamefont
  {Hastings}}, \ and\ \bibinfo {author} {\bibfnamefont {S.}~\bibnamefont
  {Michalakis}},\ }\bibfield  {title} {\enquote {\bibinfo {title} {Topological
  quantum order: stability under local perturbations},}\ }\href {\doibase
  10.1063/1.3490195} {\bibfield  {journal} {\bibinfo  {journal} {Journal of
  Mathematical Physics}\ }\textbf {\bibinfo {volume} {51}},\ \bibinfo {pages}
  {093512} (\bibinfo {year} {2010}{\natexlab{a}})},\ \Eprint
  {http://arxiv.org/abs/1001.0344} {arXiv:1001.0344} \BibitemShut {NoStop}%
\bibitem [{\citenamefont {Bravyi}\ \emph
  {et~al.}(2010{\natexlab{b}})\citenamefont {Bravyi}, \citenamefont {Poulin},\
  and\ \citenamefont {Terhal}}]{BPT}%
  \BibitemOpen
  \bibfield  {author} {\bibinfo {author} {\bibfnamefont {S.}~\bibnamefont
  {Bravyi}}, \bibinfo {author} {\bibfnamefont {D.}~\bibnamefont {Poulin}}, \
  and\ \bibinfo {author} {\bibfnamefont {B.}~\bibnamefont {Terhal}},\
  }\bibfield  {title} {\enquote {\bibinfo {title} {Tradeoffs for reliable
  quantum information storage in {2D} systems},}\ }\href {\doibase
  10.1103/PhysRevLett.104.050503} {\bibfield  {journal} {\bibinfo  {journal}
  {Phys. Rev. Lett.}\ }\textbf {\bibinfo {volume} {104}},\ \bibinfo {pages}
  {050503} (\bibinfo {year} {2010}{\natexlab{b}})},\ \Eprint
  {http://arxiv.org/abs/0909.5200} {arXiv:0909.5200} \BibitemShut {NoStop}%
\bibitem [{\citenamefont {Haah}\ and\ \citenamefont
  {Preskill}(2012)}]{HaahPreskill2012tradeoff}%
  \BibitemOpen
  \bibfield  {author} {\bibinfo {author} {\bibfnamefont {J.}~\bibnamefont
  {Haah}}\ and\ \bibinfo {author} {\bibfnamefont {J.}~\bibnamefont
  {Preskill}},\ }\bibfield  {title} {\enquote {\bibinfo {title} {Logical
  operator tradeoff for local quantum codes},}\ }\href {\doibase
  10.1103/PhysRevA.86.032308} {\bibfield  {journal} {\bibinfo  {journal} {Phys.
  Rev. A}\ }\textbf {\bibinfo {volume} {86}},\ \bibinfo {pages} {032308}
  (\bibinfo {year} {2012})},\ \Eprint {http://arxiv.org/abs/1011.3529}
  {1011.3529} \BibitemShut {NoStop}%
\bibitem [{\citenamefont {Brell}\ \emph {et~al.}(2011)\citenamefont {Brell},
  \citenamefont {Flammia}, \citenamefont {Bartlett},\ and\ \citenamefont
  {Doherty}}]{Brell2011}%
  \BibitemOpen
  \bibfield  {author} {\bibinfo {author} {\bibfnamefont {C.~G.}\ \bibnamefont
  {Brell}}, \bibinfo {author} {\bibfnamefont {S.~T.}\ \bibnamefont {Flammia}},
  \bibinfo {author} {\bibfnamefont {S.~D.}\ \bibnamefont {Bartlett}}, \ and\
  \bibinfo {author} {\bibfnamefont {A.~C.}\ \bibnamefont {Doherty}},\
  }\bibfield  {title} {\enquote {\bibinfo {title} {Toric codes and quantum
  doubles from two-body {H}amiltonians},}\ }\href {\doibase
  10.1088/1367-2630/13/5/053039} {\bibfield  {journal} {\bibinfo  {journal}
  {New Journal of Physics}\ }\textbf {\bibinfo {volume} {13}},\ \bibinfo
  {pages} {053039} (\bibinfo {year} {2011})},\ \Eprint
  {http://arxiv.org/abs/1011.1942} {arXiv:1011.1942} \BibitemShut {NoStop}%
\bibitem [{\citenamefont {Criger}\ and\ \citenamefont
  {Terhal}(2016)}]{Criger2016}%
  \BibitemOpen
  \bibfield  {author} {\bibinfo {author} {\bibfnamefont {B.}~\bibnamefont
  {Criger}}\ and\ \bibinfo {author} {\bibfnamefont {B.}~\bibnamefont
  {Terhal}},\ }\bibfield  {title} {\enquote {\bibinfo {title} {Noise thresholds
  for the [4,2,2]-concatenated toric code},}\ }\href@noop {} {\bibfield
  {journal} {\bibinfo  {journal} {Quant. Inf. Comput.}\ }\textbf {\bibinfo
  {volume} {16}},\ \bibinfo {pages} {1261} (\bibinfo {year} {2016})},\ \Eprint
  {http://arxiv.org/abs/1604.04062} {arXiv:1604.04062} \BibitemShut {NoStop}%
\bibitem [{\citenamefont {Kim}\ and\ \citenamefont
  {Kastoryano}(2017)}]{Kim2017}%
  \BibitemOpen
  \bibfield  {author} {\bibinfo {author} {\bibfnamefont {I.~H.}\ \bibnamefont
  {Kim}}\ and\ \bibinfo {author} {\bibfnamefont {M.~J.}\ \bibnamefont
  {Kastoryano}},\ }\bibfield  {title} {\enquote {\bibinfo {title} {Entanglement
  renormalization, quantum error correction, and bulk causality},}\ }\href@noop
  {} {\  (\bibinfo {year} {2017})},\ \Eprint {http://arxiv.org/abs/1701.00050}
  {arXiv:1701.00050} \BibitemShut {NoStop}%
\bibitem [{\citenamefont {Bravyi}\ and\ \citenamefont
  {Terhal}(2009)}]{BravyiTerhal2009}%
  \BibitemOpen
  \bibfield  {author} {\bibinfo {author} {\bibfnamefont {S.}~\bibnamefont
  {Bravyi}}\ and\ \bibinfo {author} {\bibfnamefont {B.}~\bibnamefont
  {Terhal}},\ }\bibfield  {title} {\enquote {\bibinfo {title} {A no-go theorem
  for a two-dimensional self-correcting quantum memory based on stabilizer
  codes},}\ }\href {\doibase 10.1088/1367-2630/11/4/043029} {\bibfield
  {journal} {\bibinfo  {journal} {New Journal of Physics}\ }\textbf {\bibinfo
  {volume} {11}},\ \bibinfo {pages} {043029} (\bibinfo {year} {2009})},\
  \Eprint {http://arxiv.org/abs/0810.1983} {arXiv:0810.1983} \BibitemShut
  {NoStop}%
\bibitem [{\citenamefont {Delfosse}(2013)}]{Delfosse2013a}%
  \BibitemOpen
  \bibfield  {author} {\bibinfo {author} {\bibfnamefont {N.}~\bibnamefont
  {Delfosse}},\ }\bibfield  {title} {\enquote {\bibinfo {title} {Tradeoffs for
  reliable quantum information storage in surface codes and color codes},}\
  }in\ \href {\doibase 10.1109/isit.2013.6620360} {\emph {\bibinfo {booktitle}
  {2013 {IEEE} International Symposium on Information Theory}}}\ (\bibinfo
  {publisher} {Institute of Electrical {\&} Electronics Engineers ({IEEE})},\
  \bibinfo {year} {2013})\ \Eprint {http://arxiv.org/abs/1301.6588}
  {arXiv:1301.6588} \BibitemShut {NoStop}%
\bibitem [{\citenamefont {Bravyi}(2011)}]{Bravyi2011b}%
  \BibitemOpen
  \bibfield  {author} {\bibinfo {author} {\bibfnamefont {S.}~\bibnamefont
  {Bravyi}},\ }\bibfield  {title} {\enquote {\bibinfo {title} {Subsystem codes
  with spatially local generators},}\ }\href {\doibase
  10.1103/PhysRevA.83.012320} {\bibfield  {journal} {\bibinfo  {journal} {Phys.
  Rev. A}\ }\textbf {\bibinfo {volume} {83}},\ \bibinfo {pages} {012320}
  (\bibinfo {year} {2011})},\ \Eprint {http://arxiv.org/abs/1008.1029}
  {arXiv:1008.1029} \BibitemShut {NoStop}%
\bibitem [{\citenamefont {Bacon}\ \emph {et~al.}(2015)\citenamefont {Bacon},
  \citenamefont {Flammia}, \citenamefont {Harrow},\ and\ \citenamefont
  {Shi}}]{Bacon2015}%
  \BibitemOpen
  \bibfield  {author} {\bibinfo {author} {\bibfnamefont {D.}~\bibnamefont
  {Bacon}}, \bibinfo {author} {\bibfnamefont {S.~T.}\ \bibnamefont {Flammia}},
  \bibinfo {author} {\bibfnamefont {A.~W.}\ \bibnamefont {Harrow}}, \ and\
  \bibinfo {author} {\bibfnamefont {J.}~\bibnamefont {Shi}},\ }\bibfield
  {title} {\enquote {\bibinfo {title} {{S}parse {Q}uantum {C}odes from
  {Q}uantum {C}ircuits},}\ }in\ \href {\doibase 10.1145/2746539.2746608} {\emph
  {\bibinfo {booktitle} {Proceedings of the Forty-Seventh Annual ACM on
  Symposium on Theory of Computing}}},\ \bibinfo {series and number} {STOC
  '15}\ (\bibinfo  {publisher} {{ACM} Press},\ \bibinfo {address} {New York,
  NY, USA},\ \bibinfo {year} {2015})\ pp.\ \bibinfo {pages} {327--334},\
  \Eprint {http://arxiv.org/abs/1411.3334} {arXiv:1411.3334} \BibitemShut
  {NoStop}%
\bibitem [{\citenamefont {Cafaro}\ and\ \citenamefont {{van
  Loock}}(2014)}]{Cafaro2014}%
  \BibitemOpen
  \bibfield  {author} {\bibinfo {author} {\bibfnamefont {C.}~\bibnamefont
  {Cafaro}}\ and\ \bibinfo {author} {\bibfnamefont {P.}~\bibnamefont {{van
  Loock}}},\ }\bibfield  {title} {\enquote {\bibinfo {title} {{Approximate
  quantum error correction for generalized amplitude-damping errors}},}\ }\href
  {\doibase 10.1103/PhysRevA.89.022316} {\bibfield  {journal} {\bibinfo
  {journal} {Phys. Rev. A}\ }\textbf {\bibinfo {volume} {89}},\ \bibinfo
  {pages} {022316} (\bibinfo {year} {2014})},\ \Eprint
  {http://arxiv.org/abs/1308.4582} {arXiv:1308.4582} \BibitemShut {NoStop}%
\bibitem [{\citenamefont {Ouyang}(2014)}]{Ouyang2014}%
  \BibitemOpen
  \bibfield  {author} {\bibinfo {author} {\bibfnamefont {Y.}~\bibnamefont
  {Ouyang}},\ }\bibfield  {title} {\enquote {\bibinfo {title}
  {Permutation-invariant quantum codes},}\ }\href {\doibase
  10.1103/PhysRevA.90.062317} {\bibfield  {journal} {\bibinfo  {journal} {Phys.
  Rev. A}\ }\textbf {\bibinfo {volume} {90}},\ \bibinfo {pages} {062317}
  (\bibinfo {year} {2014})},\ \Eprint {http://arxiv.org/abs/1302.3247}
  {arXiv:1302.3247} \BibitemShut {NoStop}%
\bibitem [{\citenamefont {Grassl}\ \emph {et~al.}(2015)\citenamefont {Grassl},
  \citenamefont {Kong}, \citenamefont {Wei}, \citenamefont {Yin},\ and\
  \citenamefont {Zeng}}]{Grassl2015}%
  \BibitemOpen
  \bibfield  {author} {\bibinfo {author} {\bibfnamefont {M.}~\bibnamefont
  {Grassl}}, \bibinfo {author} {\bibfnamefont {L.}~\bibnamefont {Kong}},
  \bibinfo {author} {\bibfnamefont {Z.}~\bibnamefont {Wei}}, \bibinfo {author}
  {\bibfnamefont {Z.-Q.}\ \bibnamefont {Yin}}, \ and\ \bibinfo {author}
  {\bibfnamefont {B.}~\bibnamefont {Zeng}},\ }\bibfield  {title} {\enquote
  {\bibinfo {title} {{Q}uantum {E}rror-{C}orrecting {C}odes for {Q}udit
  {A}mplitude {D}amping},}\ }\href@noop {} {\bibfield  {journal} {\bibinfo
  {journal} {arXiv:1509.06829}\ } (\bibinfo {year} {2015})},\ \Eprint
  {http://arxiv.org/abs/1509.06829} {arXiv:1509.06829} \BibitemShut {NoStop}%
\bibitem [{\citenamefont {Schumacher}\ and\ \citenamefont
  {Westmoreland}(2002)}]{Schumacher2002}%
  \BibitemOpen
  \bibfield  {author} {\bibinfo {author} {\bibfnamefont {B.}~\bibnamefont
  {Schumacher}}\ and\ \bibinfo {author} {\bibfnamefont {M.~D.}\ \bibnamefont
  {Westmoreland}},\ }\bibfield  {title} {\enquote {\bibinfo {title}
  {Approximate quantum error correction},}\ }\href {\doibase
  10.1023/A:1019653202562} {\bibfield  {journal} {\bibinfo  {journal} {Quant.
  Info. Process.}\ }\textbf {\bibinfo {volume} {1}},\ \bibinfo {pages} {5--12}
  (\bibinfo {year} {2002})},\ \Eprint {http://arxiv.org/abs/quant-ph/0112106}
  {quant-ph/0112106} \BibitemShut {NoStop}%
\bibitem [{\citenamefont {B\'eny}\ and\ \citenamefont
  {Oreshkov}(2010)}]{BenyOreshkov2010}%
  \BibitemOpen
  \bibfield  {author} {\bibinfo {author} {\bibfnamefont {C.}~\bibnamefont
  {B\'eny}}\ and\ \bibinfo {author} {\bibfnamefont {O.}~\bibnamefont
  {Oreshkov}},\ }\bibfield  {title} {\enquote {\bibinfo {title} {General
  conditions for approximate quantum error correction and near-optimal recovery
  channels},}\ }\href {\doibase 10.1103/PhysRevLett.104.120501} {\bibfield
  {journal} {\bibinfo  {journal} {Phys. Rev. Lett.}\ }\textbf {\bibinfo
  {volume} {104}},\ \bibinfo {pages} {120501} (\bibinfo {year} {2010})},\
  \Eprint {http://arxiv.org/abs/0907.5391} {arXiv:0907.5391} \BibitemShut
  {NoStop}%
\bibitem [{\citenamefont {Ng}\ and\ \citenamefont
  {Mandayam}(2010)}]{NgMandayam2010}%
  \BibitemOpen
  \bibfield  {author} {\bibinfo {author} {\bibfnamefont {H.~K.}\ \bibnamefont
  {Ng}}\ and\ \bibinfo {author} {\bibfnamefont {P.}~\bibnamefont {Mandayam}},\
  }\bibfield  {title} {\enquote {\bibinfo {title} {Simple approach to
  approximate quantum error correction based on the transpose channel},}\
  }\href {\doibase 10.1103/PhysRevA.81.062342} {\bibfield  {journal} {\bibinfo
  {journal} {Phys. Rev. A}\ }\textbf {\bibinfo {volume} {81}},\ \bibinfo
  {pages} {062342} (\bibinfo {year} {2010})},\ \Eprint
  {http://arxiv.org/abs/0909.0931} {arXiv:0909.0931} \BibitemShut {NoStop}%
\bibitem [{\citenamefont {Bridgeman}\ \emph {et~al.}(2016)\citenamefont
  {Bridgeman}, \citenamefont {Flammia},\ and\ \citenamefont
  {Poulin}}]{Bridgeman2016}%
  \BibitemOpen
  \bibfield  {author} {\bibinfo {author} {\bibfnamefont {J.~C.}\ \bibnamefont
  {Bridgeman}}, \bibinfo {author} {\bibfnamefont {S.~T.}\ \bibnamefont
  {Flammia}}, \ and\ \bibinfo {author} {\bibfnamefont {D.}~\bibnamefont
  {Poulin}},\ }\bibfield  {title} {\enquote {\bibinfo {title} {{D}etecting
  {T}opological {O}rder with {R}ibbon {O}perators},}\ }\href {\doibase
  10.1103/PhysRevB.94.205123} {\bibfield  {journal} {\bibinfo  {journal} {Phys.
  Rev. B}\ }\textbf {\bibinfo {volume} {94}},\ \bibinfo {pages} {205123}
  (\bibinfo {year} {2016})},\ \Eprint {http://arxiv.org/abs/1603.02275}
  {arXiv:1603.02275} \BibitemShut {NoStop}%
\bibitem [{\citenamefont {Chubb}\ and\ \citenamefont
  {Flammia}(2016)}]{Chubb2016}%
  \BibitemOpen
  \bibfield  {author} {\bibinfo {author} {\bibfnamefont {C.~T.}\ \bibnamefont
  {Chubb}}\ and\ \bibinfo {author} {\bibfnamefont {S.~T.}\ \bibnamefont
  {Flammia}},\ }\bibfield  {title} {\enquote {\bibinfo {title} {{A}pproximate
  symmetries of {H}amiltonians},}\ }\href@noop {} {\  (\bibinfo {year}
  {2016})},\ \Eprint {http://arxiv.org/abs/1608.02600} {arXiv:1608.02600}
  \BibitemShut {NoStop}%
\bibitem [{\citenamefont {Uhlmann}(1976)}]{uhlmann1976}%
  \BibitemOpen
  \bibfield  {author} {\bibinfo {author} {\bibfnamefont {A.}~\bibnamefont
  {Uhlmann}},\ }\bibfield  {title} {\enquote {\bibinfo {title} {The
  ``transition probability'' in the state space of a $*$-algebra},}\ }\href
  {\doibase 10.1016/0034-4877(76)90060-4} {\bibfield  {journal} {\bibinfo
  {journal} {Reports on Mathematical Physics}\ }\textbf {\bibinfo {volume}
  {9}},\ \bibinfo {pages} {273--279} (\bibinfo {year} {1976})}\BibitemShut
  {NoStop}%
\bibitem [{\citenamefont {Fuchs}\ and\ \citenamefont {van~de
  Graaf}(1999)}]{Fuchs1999}%
  \BibitemOpen
  \bibfield  {author} {\bibinfo {author} {\bibfnamefont {C.~A.}\ \bibnamefont
  {Fuchs}}\ and\ \bibinfo {author} {\bibfnamefont {J.}~\bibnamefont {van~de
  Graaf}},\ }\bibfield  {title} {\enquote {\bibinfo {title} {Cryptographic
  distinguishability measures for quantum-mechanical states},}\ }\href
  {\doibase 10.1109/18.761271} {\bibfield  {journal} {\bibinfo  {journal} {IEEE
  Trans. Inf. Theory}\ }\textbf {\bibinfo {volume} {45}},\ \bibinfo {pages}
  {1216} (\bibinfo {year} {1999})},\ \Eprint
  {http://arxiv.org/abs/quant-ph/9712042} {quant-ph/9712042} \BibitemShut
  {NoStop}%
\bibitem [{\citenamefont {Knill}\ and\ \citenamefont
  {Laflamme}(1997)}]{KnillLaflamme1997Theory}%
  \BibitemOpen
  \bibfield  {author} {\bibinfo {author} {\bibfnamefont {E.}~\bibnamefont
  {Knill}}\ and\ \bibinfo {author} {\bibfnamefont {R.}~\bibnamefont
  {Laflamme}},\ }\bibfield  {title} {\enquote {\bibinfo {title} {Theory of
  quantum error-correcting codes},}\ }\href {\doibase 10.1103/PhysRevA.55.900}
  {\bibfield  {journal} {\bibinfo  {journal} {Phys. Rev. A}\ }\textbf {\bibinfo
  {volume} {55}},\ \bibinfo {pages} {900--911} (\bibinfo {year}
  {1997})}\BibitemShut {NoStop}%
\bibitem [{\citenamefont {Kretschmann}\ \emph {et~al.}(2008)\citenamefont
  {Kretschmann}, \citenamefont {Schlingemann},\ and\ \citenamefont
  {Werner}}]{kretschmann2008}%
  \BibitemOpen
  \bibfield  {author} {\bibinfo {author} {\bibfnamefont {D.}~\bibnamefont
  {Kretschmann}}, \bibinfo {author} {\bibfnamefont {D.}~\bibnamefont
  {Schlingemann}}, \ and\ \bibinfo {author} {\bibfnamefont {R.~F.}\
  \bibnamefont {Werner}},\ }\bibfield  {title} {\enquote {\bibinfo {title} {The
  information-disturbance tradeoff and the continuity of {S}tinespring's
  representation},}\ }\href {\doibase 10.1109/tit.2008.917696} {\bibfield
  {journal} {\bibinfo  {journal} {IEEE Transactions on Information Theory}\
  }\textbf {\bibinfo {volume} {54}},\ \bibinfo {pages} {1708--1717} (\bibinfo
  {year} {2008})},\ \Eprint {http://arxiv.org/abs/quant-ph/0605009}
  {quant-ph/0605009} \BibitemShut {NoStop}%
\bibitem [{\citenamefont {Hayden}\ and\ \citenamefont
  {Winter}(2012)}]{hayden2012}%
  \BibitemOpen
  \bibfield  {author} {\bibinfo {author} {\bibfnamefont {P.}~\bibnamefont
  {Hayden}}\ and\ \bibinfo {author} {\bibfnamefont {A.}~\bibnamefont
  {Winter}},\ }\bibfield  {title} {\enquote {\bibinfo {title} {Weak decoupling
  duality and quantum identification},}\ }\href {\doibase
  10.1109/tit.2012.2191695} {\bibfield  {journal} {\bibinfo  {journal} {IEEE
  Transactions on Information Theory}\ }\textbf {\bibinfo {volume} {58}},\
  \bibinfo {pages} {4914--4929} (\bibinfo {year} {2012})},\ \Eprint
  {http://arxiv.org/abs/1003.4994} {arXiv:1003.4994} \BibitemShut {NoStop}%
\bibitem [{\citenamefont {Schumacher}\ and\ \citenamefont
  {Nielsen}(1996)}]{SchumacherNielsen1996Quantum}%
  \BibitemOpen
  \bibfield  {author} {\bibinfo {author} {\bibfnamefont {B.}~\bibnamefont
  {Schumacher}}\ and\ \bibinfo {author} {\bibfnamefont {M.~A.}\ \bibnamefont
  {Nielsen}},\ }\bibfield  {title} {\enquote {\bibinfo {title} {Quantum data
  processing and error correction},}\ }\href {\doibase
  10.1103/physreva.54.2629} {\bibfield  {journal} {\bibinfo  {journal} {Phys.
  Rev. A}\ }\textbf {\bibinfo {volume} {54}},\ \bibinfo {pages} {2629--2635}
  (\bibinfo {year} {1996})},\ \Eprint {http://arxiv.org/abs/quant-ph/9604022}
  {quant-ph/9604022} \BibitemShut {NoStop}%
\bibitem [{\citenamefont {Kitaev}\ \emph {et~al.}(2002)\citenamefont {Kitaev},
  \citenamefont {Shen},\ and\ \citenamefont
  {Vyalyi}}]{KitaevShenVyalyi2002CQC}%
  \BibitemOpen
  \bibfield  {author} {\bibinfo {author} {\bibfnamefont {A.~Y.}\ \bibnamefont
  {Kitaev}}, \bibinfo {author} {\bibfnamefont {A.~H.}\ \bibnamefont {Shen}}, \
  and\ \bibinfo {author} {\bibfnamefont {M.~N.}\ \bibnamefont {Vyalyi}},\
  }\href@noop {} {\emph {\bibinfo {title} {Classical and Quantum
  Computation}}}\ (\bibinfo  {publisher} {American Mathematical Society},\
  \bibinfo {year} {2002})\BibitemShut {NoStop}%
\bibitem [{\citenamefont {M\"{u}ller-Lennert}\ \emph
  {et~al.}(2013)\citenamefont {M\"{u}ller-Lennert}, \citenamefont {Dupuis},
  \citenamefont {Szehr}, \citenamefont {Fehr},\ and\ \citenamefont
  {Tomamichel}}]{MuellerLennertDupuisSzehrFehrTomamichel2013}%
  \BibitemOpen
  \bibfield  {author} {\bibinfo {author} {\bibfnamefont {M.}~\bibnamefont
  {M\"{u}ller-Lennert}}, \bibinfo {author} {\bibfnamefont {F.}~\bibnamefont
  {Dupuis}}, \bibinfo {author} {\bibfnamefont {O.}~\bibnamefont {Szehr}},
  \bibinfo {author} {\bibfnamefont {S.}~\bibnamefont {Fehr}}, \ and\ \bibinfo
  {author} {\bibfnamefont {M.}~\bibnamefont {Tomamichel}},\ }\bibfield  {title}
  {\enquote {\bibinfo {title} {On quantum {R}\'{e}nyi entropies: a new
  generalization and some properties},}\ }\href {\doibase 10.1063/1.4838856}
  {\bibfield  {journal} {\bibinfo  {journal} {J. Math. Phys.}\ }\textbf
  {\bibinfo {volume} {54}},\ \bibinfo {pages} {122203} (\bibinfo {year}
  {2013})},\ \Eprint {http://arxiv.org/abs/1306.3142} {arXiv:1306.3142}
  \BibitemShut {NoStop}%
\bibitem [{\citenamefont {Beigi}(2013)}]{Beigi2013Sandwiched}%
  \BibitemOpen
  \bibfield  {author} {\bibinfo {author} {\bibfnamefont {S.}~\bibnamefont
  {Beigi}},\ }\bibfield  {title} {\enquote {\bibinfo {title} {Sandwiched
  {R}{\'e}nyi divergence satisfies data processing inequality},}\ }\href
  {\doibase 10.1063/1.4838855} {\bibfield  {journal} {\bibinfo  {journal} {J.
  Math. Phys.}\ }\textbf {\bibinfo {volume} {54}},\ \bibinfo {pages} {122202}
  (\bibinfo {year} {2013})},\ \Eprint {http://arxiv.org/abs/1306.5920}
  {arXiv:1306.5920} \BibitemShut {NoStop}%
\bibitem [{\citenamefont {Yoshida}\ and\ \citenamefont
  {Chuang}(2010)}]{YoshidaChuang2010Framework}%
  \BibitemOpen
  \bibfield  {author} {\bibinfo {author} {\bibfnamefont {B.}~\bibnamefont
  {Yoshida}}\ and\ \bibinfo {author} {\bibfnamefont {I.~L.}\ \bibnamefont
  {Chuang}},\ }\bibfield  {title} {\enquote {\bibinfo {title} {Framework for
  classifying logical operators in stabilizer codes},}\ }\href {\doibase
  10.1103/PhysRevA.81.052302} {\bibfield  {journal} {\bibinfo  {journal} {Phys.
  Rev. A}\ }\textbf {\bibinfo {volume} {81}},\ \bibinfo {pages} {052302}
  (\bibinfo {year} {2010})},\ \Eprint {http://arxiv.org/abs/1002.0085}
  {arXiv:1002.0085} \BibitemShut {NoStop}%
\bibitem [{\citenamefont {Preskill}(1999)}]{Preskill1999QEC}%
  \BibitemOpen
  \bibfield  {author} {\bibinfo {author} {\bibfnamefont {J.}~\bibnamefont
  {Preskill}},\ }\href
  {http://www.theory.caltech.edu/~preskill/ph229/notes/chap7.pdf} {\enquote
  {\bibinfo {title} {Quantum error correction, {L}ecture notes for {P}hysics
  219, {C}altech},}\ } (\bibinfo {year} {1999})\BibitemShut {NoStop}%
\bibitem [{\citenamefont {Vafa}(1988)}]{vafa}%
  \BibitemOpen
  \bibfield  {author} {\bibinfo {author} {\bibfnamefont {C.}~\bibnamefont
  {Vafa}},\ }\bibfield  {title} {\enquote {\bibinfo {title} {Toward
  classification of conformal theories},}\ }\href {\doibase
  10.1016/0370-2693(88)91603-6} {\bibfield  {journal} {\bibinfo  {journal}
  {Physics Letters B}\ }\textbf {\bibinfo {volume} {206}},\ \bibinfo {pages}
  {421--426} (\bibinfo {year} {1988})}\BibitemShut {NoStop}%
\bibitem [{\citenamefont {Chamon}(2005)}]{Chamon2005Quantum}%
  \BibitemOpen
  \bibfield  {author} {\bibinfo {author} {\bibfnamefont {C.}~\bibnamefont
  {Chamon}},\ }\bibfield  {title} {\enquote {\bibinfo {title} {Quantum
  glassiness},}\ }\href {\doibase 10.1103/PhysRevLett.94.040402} {\bibfield
  {journal} {\bibinfo  {journal} {Phys. Rev. Lett.}\ }\textbf {\bibinfo
  {volume} {94}},\ \bibinfo {pages} {040402} (\bibinfo {year} {2005})},\
  \Eprint {http://arxiv.org/abs/cond-mat/0404182} {cond-mat/0404182}
  \BibitemShut {NoStop}%
\bibitem [{\citenamefont {Bravyi}\ \emph {et~al.}(2011)\citenamefont {Bravyi},
  \citenamefont {Leemhuis},\ and\ \citenamefont
  {Terhal}}]{BravyiLeemhuisTerhal2011Topological}%
  \BibitemOpen
  \bibfield  {author} {\bibinfo {author} {\bibfnamefont {S.}~\bibnamefont
  {Bravyi}}, \bibinfo {author} {\bibfnamefont {B.}~\bibnamefont {Leemhuis}}, \
  and\ \bibinfo {author} {\bibfnamefont {B.~M.}\ \bibnamefont {Terhal}},\
  }\bibfield  {title} {\enquote {\bibinfo {title} {Topological order in an
  exactly solvable 3{D} spin model},}\ }\href {\doibase
  10.1016/j.aop.2010.11.002} {\bibfield  {journal} {\bibinfo  {journal} {Annals
  of Physics}\ }\textbf {\bibinfo {volume} {326}},\ \bibinfo {pages} {839--866}
  (\bibinfo {year} {2011})},\ \Eprint {http://arxiv.org/abs/1006.4871}
  {arXiv:1006.4871} \BibitemShut {NoStop}%
\bibitem [{\citenamefont {Haah}(2011)}]{Haah2011Local}%
  \BibitemOpen
  \bibfield  {author} {\bibinfo {author} {\bibfnamefont {J.}~\bibnamefont
  {Haah}},\ }\bibfield  {title} {\enquote {\bibinfo {title} {Local stabilizer
  codes in three dimensions without string logical operators},}\ }\href
  {\doibase 10.1103/PhysRevB.89.075119} {\bibfield  {journal} {\bibinfo
  {journal} {Phys. Rev. A}\ }\textbf {\bibinfo {volume} {83}},\ \bibinfo
  {pages} {042330} (\bibinfo {year} {2011})},\ \Eprint
  {http://arxiv.org/abs/1310.4507} {arXiv:1310.4507} \BibitemShut {NoStop}%
\bibitem [{\citenamefont {Vijay}\ \emph {et~al.}(2015)\citenamefont {Vijay},
  \citenamefont {Haah},\ and\ \citenamefont {Fu}}]{VijayHaahFu2015New}%
  \BibitemOpen
  \bibfield  {author} {\bibinfo {author} {\bibfnamefont {S.}~\bibnamefont
  {Vijay}}, \bibinfo {author} {\bibfnamefont {J.}~\bibnamefont {Haah}}, \ and\
  \bibinfo {author} {\bibfnamefont {L.}~\bibnamefont {Fu}},\ }\bibfield
  {title} {\enquote {\bibinfo {title} {A new kind of topological quantum order:
  A dimensional hierarchy of quasiparticles built from stationary
  excitations},}\ }\href {\doibase 10.1103/PhysRevB.92.235136} {\bibfield
  {journal} {\bibinfo  {journal} {Phys. Rev. B}\ }\textbf {\bibinfo {volume}
  {92}},\ \bibinfo {pages} {235136} (\bibinfo {year} {2015})},\ \Eprint
  {http://arxiv.org/abs/1505.02576} {1505.02576} \BibitemShut {NoStop}%
\bibitem [{\citenamefont {Haah}(2013)}]{Haah2012PauliModule}%
  \BibitemOpen
  \bibfield  {author} {\bibinfo {author} {\bibfnamefont {J.}~\bibnamefont
  {Haah}},\ }\bibfield  {title} {\enquote {\bibinfo {title} {Commuting {P}auli
  {H}amiltonians as maps between free modules},}\ }\href {\doibase
  10.1007/s00220-013-1810-2} {\bibfield  {journal} {\bibinfo  {journal}
  {Commun. Math. Phys.}\ }\textbf {\bibinfo {volume} {324}},\ \bibinfo {pages}
  {351--399} (\bibinfo {year} {2013})},\ \Eprint
  {http://arxiv.org/abs/1204.1063} {arXiv:1204.1063} \BibitemShut {NoStop}%
\bibitem [{\citenamefont {Hastings}(2004)}]{hastings2004}%
  \BibitemOpen
  \bibfield  {author} {\bibinfo {author} {\bibfnamefont {M.~B.}\ \bibnamefont
  {Hastings}},\ }\bibfield  {title} {\enquote {\bibinfo {title}
  {{Lieb-Schultz-Mattis} in higher dimensions},}\ }\href {\doibase
  10.1103/physrevb.69.104431} {\bibfield  {journal} {\bibinfo  {journal}
  {Physical Review B}\ }\textbf {\bibinfo {volume} {69}},\ \bibinfo {pages}
  {104431} (\bibinfo {year} {2004})},\ \Eprint
  {http://arxiv.org/abs/cond-mat/0305505} {cond-mat/0305505} \BibitemShut
  {NoStop}%
\bibitem [{\citenamefont {Bachmann}\ \emph {et~al.}(2012)\citenamefont
  {Bachmann}, \citenamefont {Michalakis}, \citenamefont {Nachtergaele},\ and\
  \citenamefont {Sims}}]{bachmann2012}%
  \BibitemOpen
  \bibfield  {author} {\bibinfo {author} {\bibfnamefont {S.}~\bibnamefont
  {Bachmann}}, \bibinfo {author} {\bibfnamefont {S.}~\bibnamefont
  {Michalakis}}, \bibinfo {author} {\bibfnamefont {B.}~\bibnamefont
  {Nachtergaele}}, \ and\ \bibinfo {author} {\bibfnamefont {R.}~\bibnamefont
  {Sims}},\ }\bibfield  {title} {\enquote {\bibinfo {title} {Automorphic
  equivalence within gapped phases of quantum lattice systems},}\ }\href
  {\doibase 10.1007/s00220-011-1380-0} {\bibfield  {journal} {\bibinfo
  {journal} {Communications in Mathematical Physics}\ }\textbf {\bibinfo
  {volume} {309}},\ \bibinfo {pages} {835--871} (\bibinfo {year} {2012})},\
  \Eprint {http://arxiv.org/abs/1102.0842} {arXiv:1102.0842} \BibitemShut
  {NoStop}%
\bibitem [{\citenamefont {Eastin}\ and\ \citenamefont
  {Knill}(2009)}]{eastinknill}%
  \BibitemOpen
  \bibfield  {author} {\bibinfo {author} {\bibfnamefont {B.}~\bibnamefont
  {Eastin}}\ and\ \bibinfo {author} {\bibfnamefont {E.}~\bibnamefont {Knill}},\
  }\bibfield  {title} {\enquote {\bibinfo {title} {Restrictions on transversal
  encoded quantum gate sets},}\ }\href {\doibase
  10.1103/PhysRevLett.102.110502} {\bibfield  {journal} {\bibinfo  {journal}
  {Physical Review Letters}\ }\textbf {\bibinfo {volume} {102}},\ \bibinfo
  {pages} {110502} (\bibinfo {year} {2009})},\ \Eprint
  {http://arxiv.org/abs/0811.4262} {arXiv:0811.4262} \BibitemShut {NoStop}%
\bibitem [{\citenamefont {Bravyi}\ and\ \citenamefont
  {K\"onig}(2013)}]{BravyiKonig2013}%
  \BibitemOpen
  \bibfield  {author} {\bibinfo {author} {\bibfnamefont {S.}~\bibnamefont
  {Bravyi}}\ and\ \bibinfo {author} {\bibfnamefont {R.}~\bibnamefont
  {K\"onig}},\ }\bibfield  {title} {\enquote {\bibinfo {title} {Classification
  of topologically protected gates for local stabilizer codes},}\ }\href
  {\doibase 10.1103/PhysRevLett.110.170503} {\bibfield  {journal} {\bibinfo
  {journal} {Phys. Rev. Lett.}\ }\textbf {\bibinfo {volume} {110}},\ \bibinfo
  {pages} {170503} (\bibinfo {year} {2013})},\ \Eprint
  {http://arxiv.org/abs/1206.1609} {arXiv:1206.1609} \BibitemShut {NoStop}%
\bibitem [{\citenamefont {Poulin}(2005)}]{Poulin2005Stabilizer}%
  \BibitemOpen
  \bibfield  {author} {\bibinfo {author} {\bibfnamefont {D.}~\bibnamefont
  {Poulin}},\ }\bibfield  {title} {\enquote {\bibinfo {title} {Stabilizer
  formalism for operator quantum error correction},}\ }\href {\doibase
  10.1103/PhysRevLett.95.230504} {\bibfield  {journal} {\bibinfo  {journal}
  {Phys. Rev. Lett.}\ }\textbf {\bibinfo {volume} {95}},\ \bibinfo {pages}
  {230504} (\bibinfo {year} {2005})},\ \Eprint
  {http://arxiv.org/abs/quant-ph/0508131} {quant-ph/0508131} \BibitemShut
  {NoStop}%
\bibitem [{\citenamefont {B{\'{e}}ny}(2009)}]{Beny2009}%
  \BibitemOpen
  \bibfield  {author} {\bibinfo {author} {\bibfnamefont {C.}~\bibnamefont
  {B{\'{e}}ny}},\ }\bibfield  {title} {\enquote {\bibinfo {title} {Conditions
  for the approximate correction of algebras},}\ }in\ \href {\doibase
  10.1007/978-3-642-10698-9_7} {\emph {\bibinfo {booktitle} {Theory of Quantum
  Computation, Communication, and Cryptography}}}\ (\bibinfo  {publisher}
  {Springer},\ \bibinfo {year} {2009})\ pp.\ \bibinfo {pages} {66--75},\
  \Eprint {http://arxiv.org/abs/0907.4207} {arXiv:0907.4207} \BibitemShut
  {NoStop}%
\bibitem [{\citenamefont {Dennis}\ \emph {et~al.}(2002)\citenamefont {Dennis},
  \citenamefont {Kitaev}, \citenamefont {Landahl},\ and\ \citenamefont
  {Preskill}}]{dennis}%
  \BibitemOpen
  \bibfield  {author} {\bibinfo {author} {\bibfnamefont {E.}~\bibnamefont
  {Dennis}}, \bibinfo {author} {\bibfnamefont {A.}~\bibnamefont {Kitaev}},
  \bibinfo {author} {\bibfnamefont {A.}~\bibnamefont {Landahl}}, \ and\
  \bibinfo {author} {\bibfnamefont {J.}~\bibnamefont {Preskill}},\ }\bibfield
  {title} {\enquote {\bibinfo {title} {Topological quantum memory},}\ }\href
  {\doibase 10.1063/1.1499754} {\bibfield  {journal} {\bibinfo  {journal}
  {Journal of Mathematical Physics}\ }\textbf {\bibinfo {volume} {43}},\
  \bibinfo {pages} {4452--4505} (\bibinfo {year} {2002})},\ \Eprint
  {http://arxiv.org/abs/quant-ph/0110143} {quant-ph/0110143} \BibitemShut
  {NoStop}%
\bibitem [{\citenamefont {Alicki}\ \emph {et~al.}(2010)\citenamefont {Alicki},
  \citenamefont {Horodecki}, \citenamefont {Horodecki},\ and\ \citenamefont
  {Horodecki}}]{alicki}%
  \BibitemOpen
  \bibfield  {author} {\bibinfo {author} {\bibfnamefont {R.}~\bibnamefont
  {Alicki}}, \bibinfo {author} {\bibfnamefont {M.}~\bibnamefont {Horodecki}},
  \bibinfo {author} {\bibfnamefont {P.}~\bibnamefont {Horodecki}}, \ and\
  \bibinfo {author} {\bibfnamefont {R.}~\bibnamefont {Horodecki}},\ }\bibfield
  {title} {\enquote {\bibinfo {title} {On thermal stability of topological
  qubit in kitaev's 4d model},}\ }\href {\doibase 10.1142/S1230161210000023}
  {\bibfield  {journal} {\bibinfo  {journal} {Open Systems \& Information
  Dynamics}\ }\textbf {\bibinfo {volume} {17}},\ \bibinfo {pages} {1--20}
  (\bibinfo {year} {2010})},\ \Eprint {http://arxiv.org/abs/0811.0033}
  {arXiv:0811.0033} \BibitemShut {NoStop}%
\bibitem [{\citenamefont {Brown}\ \emph {et~al.}(2016)\citenamefont {Brown},
  \citenamefont {Loss}, \citenamefont {Pachos}, \citenamefont {Self},\ and\
  \citenamefont {Wootton}}]{brown}%
  \BibitemOpen
  \bibfield  {author} {\bibinfo {author} {\bibfnamefont {B.~J.}\ \bibnamefont
  {Brown}}, \bibinfo {author} {\bibfnamefont {D.}~\bibnamefont {Loss}},
  \bibinfo {author} {\bibfnamefont {J.~K.}\ \bibnamefont {Pachos}}, \bibinfo
  {author} {\bibfnamefont {C.~N.}\ \bibnamefont {Self}}, \ and\ \bibinfo
  {author} {\bibfnamefont {J.~R.}\ \bibnamefont {Wootton}},\ }\bibfield
  {title} {\enquote {\bibinfo {title} {Quantum memories at finite
  temperature},}\ }\href {\doibase 10.1103/RevModPhys.88.045005} {\bibfield
  {journal} {\bibinfo  {journal} {Rev. Mod. Phys.}\ }\textbf {\bibinfo {volume}
  {88}},\ \bibinfo {pages} {045005} (\bibinfo {year} {2016})},\ \Eprint
  {http://arxiv.org/abs/1411.6643} {arXiv:1411.6643} \BibitemShut {NoStop}%
\bibitem [{\citenamefont {Pastawski}\ \emph {et~al.}(2015)\citenamefont
  {Pastawski}, \citenamefont {Yoshida}, \citenamefont {Harlow},\ and\
  \citenamefont {Preskill}}]{HAPPY}%
  \BibitemOpen
  \bibfield  {author} {\bibinfo {author} {\bibfnamefont {F.}~\bibnamefont
  {Pastawski}}, \bibinfo {author} {\bibfnamefont {B.}~\bibnamefont {Yoshida}},
  \bibinfo {author} {\bibfnamefont {D.}~\bibnamefont {Harlow}}, \ and\ \bibinfo
  {author} {\bibfnamefont {J.}~\bibnamefont {Preskill}},\ }\bibfield  {title}
  {\enquote {\bibinfo {title} {Holographic quantum error-correcting codes: Toy
  models for the bulk/boundary correspondence},}\ }\href {\doibase
  10.1007/JHEP06(2015)149} {\bibfield  {journal} {\bibinfo  {journal} {Journal
  of High Energy Physics}\ }\textbf {\bibinfo {volume} {2015}},\ \bibinfo
  {pages} {149} (\bibinfo {year} {2015})},\ \Eprint
  {http://arxiv.org/abs/1503.06237} {arXiv:1503.06237} \BibitemShut {NoStop}%
\bibitem [{\citenamefont {Harlow}(2016)}]{Harlow2016}%
  \BibitemOpen
  \bibfield  {author} {\bibinfo {author} {\bibfnamefont {D.}~\bibnamefont
  {Harlow}},\ }\bibfield  {title} {\enquote {\bibinfo {title} {The
  {R}yu-{T}akayanagi formula from quantum error correction},}\ }\href@noop {}
  {\  (\bibinfo {year} {2016})},\ \Eprint {http://arxiv.org/abs/1607.03901}
  {arXiv:1607.03901} \BibitemShut {NoStop}%
\bibitem [{\citenamefont {Audenaert}(2007)}]{Audenaert2007sharp}%
  \BibitemOpen
  \bibfield  {author} {\bibinfo {author} {\bibfnamefont {K.~M.~R.}\
  \bibnamefont {Audenaert}},\ }\bibfield  {title} {\enquote {\bibinfo {title}
  {A sharp continuity estimate for the von {N}eumann entropy},}\ }\href
  {\doibase 10.1088/1751-8113/40/28/s18} {\bibfield  {journal} {\bibinfo
  {journal} {Journal of Physics A: Mathematical and Theoretical}\ }\textbf
  {\bibinfo {volume} {40}},\ \bibinfo {pages} {8127--8136} (\bibinfo {year}
  {2007})},\ \Eprint {http://arxiv.org/abs/quant-ph/0610146} {quant-ph/0610146}
  \BibitemShut {NoStop}%
\bibitem [{\citenamefont {Alicki}\ and\ \citenamefont
  {Fannes}(2004)}]{AlickiFannes2004Continuity}%
  \BibitemOpen
  \bibfield  {author} {\bibinfo {author} {\bibfnamefont {R.}~\bibnamefont
  {Alicki}}\ and\ \bibinfo {author} {\bibfnamefont {M.}~\bibnamefont
  {Fannes}},\ }\bibfield  {title} {\enquote {\bibinfo {title} {Continuity of
  quantum conditional information},}\ }\href {\doibase
  10.1088/0305-4470/37/5/l01} {\bibfield  {journal} {\bibinfo  {journal} {J.
  Phys. A: Math. Gen.}\ }\textbf {\bibinfo {volume} {37}},\ \bibinfo {pages}
  {L55--L57} (\bibinfo {year} {2004})},\ \Eprint
  {http://arxiv.org/abs/quant-ph/0312081} {quant-ph/0312081} \BibitemShut
  {NoStop}%
\end{thebibliography}%

\appendix


\section{Proof of Theorem~\ref{thm:infodist}: Information-disturbance tradeoff}
\label{app:A}

\begin{thm}[More general version of Theorem~\ref{thm:infodist}]
Let $\Pi$ be a subspace on $QC$, and $R$ a purifying space for $\Pi$.
Given a Stinespring purification $V^{QE}$
of a channel $\cN : \rho \mapsto \Tr_E( V^{QE} \rho V^{QE\dagger})$
for $\rho = \Pi \rho = \rho \Pi$,
let $\cN^c : \rho \mapsto \Tr_Q( V^{QE} \rho V^{QE\dagger})$
be the complementary channel.
Similarly, fix a purification $W^{QE}$ of a channel $\cM$ on $Q$,
and define the complementary channel $\cM^c$.
Then, we have
\begin{align}
 \inf_{\cR} \sup_{\rho^{QCR}} \bures \left( \cM( \rho^{QCR} ), ~\cR \circ \cN(\rho^{QCR}) \right)
 =
 \inf_{\cS} \sup_{\rho^{QCR}} \bures \left( \cS \circ \cM^c (\rho^{QCR}),~ \cN^c(\rho^{QCR}) \right).
\label{eq:disturbanceTradeoff}
\end{align}
\end{thm}

This theorem with the subsystem $C$ being empty ($C = \mathbb C$)
was stated in Ref.~\cite{BenyOreshkov2010}.
The proof in Ref.~\cite{BenyOreshkov2010}
omits the last step to replace an operator of norm at most 1
with a unitary operator~\cite{kretschmann2008}.
Our statement is only different from that of Ref.~\cite{BenyOreshkov2010}
in that it includes a subsystem $C$ explicitly.
This enables us to accommodate locality.
The domain and the codomain of the channels $\cN, \cM$ being the same
is mere convenience of presentation;
more general cases reduce to the present one.
The theorem implies in particular that
when a code is defined on a physical system $ABC$ ($Q=AB$),
if $\cN = \Tr_A$, $\cN^c = \Tr_B$, $\cM = \mathrm{id}$, and $\cM^c$ outputs a fixed state $\omega$,
then
\begin{align}
 \inf_{\cR_B^{AB}} \sup_{\rho^{ABCR}} \bures \left( \rho^{ABCR},~ \cR_B^{AB}(\rho^{BCR}) \right)
 =
 \inf_{\omega^A} \sup_{\rho^{ABCR}} \bures \left( \omega^A \otimes \rho^{CR},~ \rho^{ACR} \right),
\end{align}
which is Theorem~\ref{thm:infodist}.

\begin{proof}
Let $X^{QE'}$ denote a purification of the channel $\cR$.
The new environment $E'$ is arbitrary here,
unlike the subsystem $E$ that is fixed by our choice of the complementary channel.
Any channel (CPTP map) with $d_i$-dimensional input and $d_o$-dimensional output
can be represented with a $(d_i d_o)$-dimensional environment,
which implies that the domain of all channels $\cR$ and $\cS$ is compact.
However, for the present proof, it is useful to note that the optimizations
over the channels $\cR$, $\cS$ are with arbitrary environments.

Since the fidelity $\fidelity = 1 - \bures^2$
is concave in the arguments, the maximization over $\rho^{QCR}$
can be restricted to pure states $\rho^{QCR}$.
Assuming $\rho^{QCR}$ is pure,
let us express the fidelity on the left-hand side of Eq.~\eqref{eq:disturbanceTradeoff}
using Uhlmann's theorem.
\begin{align}
LHS
&=
\sup_{X^{QE'}} \inf_{\rho^{QCR}} \sup_{Y^{EE'}}
\bra{\rho^{QCR}0^{EE'}} W^{QE\dagger} Y^{EE'} X^{QE'} V^{QE} \ket{\rho^{QCR} 0^{EE'}}\\
&=
\sup_{X^{QE'}} \inf_{\rho^Q} \sup_{Y^{EE'}}
\Tr\left[ \rho^Q 0^{EE'} W^{QE\dagger} Y^{EE'} X^{QE'} V^{QE} \right]
\end{align}
In the second line, the domain of $\rho^Q$ is convex.
(Indeed, $t \rho^Q_1 + (1-t)\rho^Q_2$
is reduced from $t \rho^{QC}_1 + (1-t) \rho^{QC}_2$,
where $\rho^{QC}_{1,2}$ are some code states from the definition of $\rho^Q_{1,2}$,
which can be purified using $R$.)
Even if we relax the domain of $Y^{QE'}$ to those of the operator $Z^{QE'}$ of norm at most 1,
the inner-most supremum always occurs when $Z^{QE'}$ is a unitary.
This follows because $\sup_{U:\norm{U}\le1} \Re\Tr AU = \onenorm{A}$ for any matrix $A$
where the supremum is achieved by a unitary from the singular value decomposition of $A$.
After such a relaxation, we can apply the minimax theorem to the bilinear objective function
to obtain
\begin{align}
LHS
&=
\sup_{X^{QE'}} \sup_{Z^{EE'}} \inf_{\rho^Q}
\Tr\left[ \rho^Q 0^{EE'} W^{QE\dagger} Z^{EE'} X^{QE'} V^{QE} \right]\\
&=
\sup_{X^{QE'}} \sup_{Z^{EE'}} \inf_{\rho^Q}
\Tr\left[ \rho^Q 0^{E} W^{QE\dagger} \bra{0^{E'}} Z^{EE'} X^{QE'} \ket{0^{E'}} V^{QE} \right].\label{eq:maximin}
\end{align}
The second line is a rearrangement to make the linear operator
$\bra{0^{E'}} Z^{EE'} X^{QE'} \ket{0^{E'}}: QE \to QE$
more visible.
Since $Z^{EE'}$ has norm at most 1,
we can introduce another subsystem $E'' = \mathbb C^2$
and find a unitary $S^{EE'E''}$
such that $Z^{EE'} = \bra{0^{E''}} S^{EE'E''} \ket{0^{E''}}$
by Lemma~\ref{lem:extensionToUnitary} below.
Then we can replace $\bra{0^{E'}} Z^{EE'} X^{QE'} \ket{0^{E'}}$
with
$\bra{0^{E'E''}} S^{QE'E''} (X^{QE'}I^{E''}) \ket{0^{E'E''}}$.
Treating $E'E''$ as a larger environment $E'$, which was arbitrary,
we see that the maximin value of \eqref{eq:maximin}
with arbitrary $E'$ can be achieved
by \emph{unitary} $Z^{EE'}$ and $X^{QE'}$.

The RHS of the theorem can be manipulated symmetrically
with the role of $X^{QE'}$ and $Y^{EE'}$ interchanged,
eventually leading to \eqref{eq:maximin} exactly.
The completes the proof of the theorem.
\end{proof}

Using the Fuchs-van de Graaf relation
\[
 \fidelity^2 + \trdist^2 \le 1 \le \fidelity + \trdist,
\]
we recover Theorem~3 of Ref.~\cite{kretschmann2008}:
\begin{corollary}
For any (noise) channel $\cN$, we have
\begin{align}
\frac14 \inf_{\cR} \norm{\cR \circ \cN - \id}^2_\diamond
\le
\inf_{\mathcal P} \norm{\cN^c - \mathcal P}_\diamond
\le
2 \inf_{\cR} \norm{\cR \circ \cN - \id}^{1/2}_\diamond
\end{align}
where $\mathcal P$ denotes any completely depolarizing channel.
\end{corollary}

\begin{lem}
For any matrix $M^A$ such that $\norm{M} \le 1$,
there exists a unitary $V^{AB}$
such that $\bra{0^B} V^{AB} \ket{0^B} = M^A$ where $B = \mathbb C^2$.
\label{lem:extensionToUnitary}
\end{lem}
\begin{proof}
Let $M = XDY$ be a singular value decomposition of $M$ with $X,Y$ being unitary.
If we had the lemma proved for the diagonal $D$,
then multiplying $X^A \otimes I^B$ and $Y^A \otimes I^B$ will prove the lemma
in the general case. Hence, it suffices to prove the lemma for any non-negative diagonal $M$.

By assumption, the non-negative diagonal elements of $M$ are $\le 1$.
The operator $V^{AB}$ has four blocks $\bra{i^B} V^{AB} \ket{j^B}$
where $i,j=0,1$.
Define $\bra{0^B} V^{AB} \ket{0^B} = M$,
and $\bra{1^B} V^{AB} \ket{0^B} = \sqrt{1-M^2}$.
Then the columns of $V^{AB}\ket{0^B}$ are orthonormal.
There exists a basis of the vector space on $AB$ that extends $V^{AB} \ket{0^B}$,
which can be used to fill the block $V^{AB} \ket{1^B}$.
\end{proof}

\section{Proof of Theorem~\ref{thm:decoupling}: Decoupling-Correctability}\label{app:B}

\begin{lem}
Two numbers
\begin{align}
\mu &= \sup_{\rho^{ABCR}} \bures( \rho^{ACR},~ \rho^A \otimes \rho^{CR} ),\\
\delta&=\inf_{\omega^A} \sup_{\rho^{ABCR}} \bures( \rho^{ACR},~ \omega^A \otimes \rho^{CR} ).
\end{align}
satisfy
\begin{align}
\frac{1}{9} \delta^2 \le \mu \le 2 \delta.
\end{align}
\end{lem}
\begin{proof}
That $\mu \le 2 \delta$ follows because $\bures( \rho^A,~\omega^A) \le \delta$ by the monotonicity of the
Bures distance.
For the other inequality, we will show
\begin{align}
\onenorm{ \rho^A - \sigma^A } \le 2(2+\sqrt 2) \mu
\end{align}
for any code states $\rho^{ABCR}$ and $\sigma^{ABCR}$.
This implies that $\bures( \rho^A,~\sigma^A ) \le \sqrt{(2+\sqrt2)\mu}$,
and therefore
\begin{align}
\delta
&\le \sup_{\rho^{ABCR}} \bures( \rho^{ACR}, \sigma^A \otimes \rho^{CR} ) \\
&\le \bures( \rho^{ACR},~\rho^A \otimes \rho^{CR}) + \sqrt{(2+\sqrt2) \mu} \\
&\le \left(1+\sqrt{2+\sqrt2}\right) \sqrt \mu .
\end{align}

Consider $\sup_{\rho,\sigma} \onenorm{ \rho^A - \sigma^A } $.
Since for every $\rho^A$ there is a code state $\rho^{ABC}$
that reduces to $\rho^A$, we may regard the domain of optimization
to be the set of mixed code states $\rho^{ABC}$.
Since the trace distance is convex in its both arguments,
the supremum occurs at extreme points, which are pure states.
So, we only need to consider the case
$\ket{\rho^{ABC}} = \sqrt{a} \ket{\sigma^{ABC}} + \sqrt{1-a} \ket{\eta^{ABC}}$
where $\Braket{\sigma^{ABC} | \eta^{ABC}}=0$ and $a \in [0,1]$.

Define
\begin{align}
\ket{\psi^{ABCR}} := \frac{1}{\sqrt{2}}\left( \ket{\sigma^{ABC}0^R} + \ket{\eta^{ABC}1^R}\right).
\end{align}
Then,
\begin{align}
\Delta := \psi^{AR} - \psi^A \otimes \psi^R =
\frac12
\begin{pmatrix}
(\sigma^A-\eta^A)/2 & \underbrace{\Tr_{BC} \ket{\sigma^{ABC}}\bra{\eta^{ABC}}}_{Y}\\
\underbrace{\Tr_{BC} \ket{\eta^{ABC}}\bra{\sigma^{ABC}}}_{Y^\dagger} & -(\sigma^A-\eta^A)/2
\end{pmatrix}.
\end{align}
Since $\onenorm{X} = \sup_U \Re\Tr(XU)$,
we see
\begin{align}
2\sqrt2 \mu
&\ge \sup_U \Re\Tr\left[ \Delta \begin{pmatrix}
U & 0 \\ 0 & -U
\end{pmatrix}\right]\\
&= \frac12 \sup_U \Re\Tr[ (\sigma^A-\eta^A)U ] = \frac 12 \onenorm{\sigma^A-\eta^A},
\end{align}
and also
\begin{align}
2\sqrt2 \mu
&\ge \sup_U \Re\Tr\left[ \Delta \begin{pmatrix}
0 & U \\ U & 0
\end{pmatrix}\right]
= \frac12 \onenorm{Y+Y^\dagger}.
\end{align}
Therefore,
\begin{align}
\onenorm{\rho^A - \sigma^A}
&=
\onenorm{(a-1) \sigma^A + (1-a)\eta^A + \sqrt{a(1-a)} (Y+Y^\dagger)}\\
&\le
(1-a)\onenorm{\sigma^A - \eta^A} + \sqrt{a(1-a)} \onenorm{Y+Y^\dagger}\\
&\le
4\sqrt2 (1-a)\mu + 4 \sqrt 2 \sqrt{a(1-a)}\mu\\
&\le
2 (2+\sqrt{2})\mu .
\end{align}
\end{proof}

\section{Proof of Theorem \ref{thm:cleaning1}: Cleaning Lemma and Correctability}\label{app:C}

Let $R$ be a purifying system for the code space $\Pi$.
\begin{lem}[A local error correction map is a local cleaning map.]
Let $A$ be a region of the physical system $ABC$.
Suppose there exists a channel $\cR_B^{AB}$ such that
\begin{align}
\sup_{\rho^{ABC}} \onenorm{\cR_B^{AB} ( \rho^{BC}) - \rho^{ABC}} \le \epsilon.
\end{align}
Then, for any operator $U^{ABC}$ that is not necessarily logical,
the pull-back $V^{BC} = (\cR_B^{AB})^*(U^{ABC})$ satisfies
\begin{align}
\norm{V^{BC}} &\le \norm{U^{ABC}} \\
\norm{\Pi (U^{ABC} - V^{BC})\Pi} &\le 2\epsilon \norm{U^{ABC}}. \label{eq:local-diff-res-codespace}
\end{align}
In addition, if $U^{ABC}$ is a logical unitary, then $V^{BC}= \cR^*(U^{ABC})$ satisfies
\begin{align}
\norm{V^{BC}} &\le 1\\
\norm{(U^{ABC} - V^{BC})\Pi} &\le 4 \sqrt \epsilon \label{eq:local-diff-unres-codomain} \\
\norm{(U^{ABC\dagger} - V^{BC\dagger}) \Pi} & \le 4 \sqrt \epsilon \label{eq:local-diff-unres-codomain2}
\end{align}
where $\Pi$ is the projector onto the code space.
\end{lem}
\begin{proof}
Since $V^{BC} = (W_B^{AB})^\dagger U^{ABC} W_B^{AB}$ for some Stinespring isometry $W_B^{AB}$,
it follows that $\norm{V^{BC}} \le \norm{U^{ABC}}$.
Whenever $U^{ABC}$ is hermitian, $V^{BC}$ is also hermitian.
Since $V^{BC\dagger} = (\cR_B^{AB})^*(U^{ABC\dagger})$,
\eqref{eq:local-diff-unres-codomain2} will follow from \eqref{eq:local-diff-unres-codomain}.

First, we estimate the norm restricted to the code space
assuming $U^{ABC}$ is hermitian.
\begin{align}
\norm{\Pi(U^{ABC} - V^{BC})\Pi}
&= \sup_{\rho^{ABC}} \abs{\Tr ( \rho^{ABC}(U^{ABC} - V^{BC}) )}\\
&= \sup_{\rho^{ABC}} \abs{\Tr ( \rho^{ABC}U^{ABC} - \cR_B^{AB}(\rho^{BC}) U^{ABC} )}\\
&= \sup_{\rho^{ABC}} \abs{\Tr ( (\rho^{ABC}-\cR_B^{AB}(\rho^{BC})) U^{ABC} )}\\
&\le \sup_{\rho^{ABC}} \onenorm{\rho^{ABC}-\cR_B^{AB}(\rho^{BC})} \norm{U^{ABC}}\\
&\le \epsilon \norm{U^{ABC}}.
\end{align}
For general $U^{ABC}$, decompose $U^{ABC}$ into hermitian and anti-hermitian parts,
and use the triangle inequality for the operator norm.
This proves Eq.~\eqref{eq:local-diff-res-codespace}.

To prove Eq.~\eqref{eq:local-diff-unres-codomain},
let $\ket \psi = \Pi \ket \psi$ be any normalized code state.
If $\alpha = \twonorm{ \Pi V^{BC} \ket \psi }$ and
$\beta = \twonorm{ (I-\Pi) V^{BC} \ket \psi }$,
then $\alpha \ge 1- 2 \epsilon$ by Eq.~\eqref{eq:local-diff-res-codespace}
and the unitarity of $U^{ABC}$.
Since $\norm{V^{BC}} \le \norm{U^{ABC}} = 1$,
we see $\alpha^2 + \beta^2 \le 1$.
Hence,
$\beta^2 \le 1 - \alpha^2 \le 1- (1-2\epsilon)^2 \le 4\epsilon$,
and
\begin{align}
\norm{(I-\Pi)V^{BC}\Pi} \le 2 \sqrt \epsilon.
\end{align}
Therefore,
\begin{align}
\norm{(U^{ABC}-V^{BC})\Pi}
&\le \norm{(I-\Pi)(U^{ABC}-V^{BC})\Pi} + \norm{\Pi (U^{ABC}-V^{BC})\Pi }\\
&\le 2 \sqrt \epsilon + 2\epsilon \\
&\le 4 \sqrt \epsilon
\end{align}
\end{proof}

\begin{lem}[A cleanable region is correctable.]
Suppose for any unitary $U^{AB}$ that preserves the code space
there exists an operator $V^B$ supported on $B$ such that
\begin{align}
\norm{V^B} &\le 1,\\
\norm{(U^{AB} - V^B) \Pi} &\le \epsilon, \\
\norm{(U^{AB\dagger} - V^{B\dagger}) \Pi} &\le \epsilon.
\end{align}
Then, for any code state $\rho^{ABR}$ with purifying space $R$
\begin{align}
 \onenorm{\rho^{AR} - \omega^A \otimes \rho^R} \le 5 \epsilon .
\end{align}
\end{lem}
\begin{proof}
Using the triangle inequality and the fact that
$\onenorm{PQ},\onenorm{QP} \le \onenorm{P}\cdot \norm{Q}$,
we have
\begin{align}
\onenorm{U^{AB} \rho^{ABR} U^{AB\dagger} - V^B \rho^{ABR} V^{B\dagger}}
&\le
\onenorm{
U^{AB} \rho^{ABR} (U^{AB\dagger} - V^{B\dagger})
}\\&\quad\quad+\onenorm{
(U^{AB}  - V^B) \rho^{ABR} V^{B\dagger}
}\\
&\le 2 \epsilon ,\\
\norm{V^{B\dagger} V^B \Pi - \Pi}
&=
\norm{V^{B\dagger} \Pi V^B \Pi - U^{AB\dagger} \Pi U^{AB} \Pi + V^{B\dagger} (I-\Pi) V^B \Pi}\\
&\le
\norm{(V^{B\dagger}-U^{AB\dagger}) \Pi V^B \Pi + U^{AB\dagger} \Pi(V^B- U^{AB}) \Pi}\\
&\quad\quad +
\norm{V^{B\dagger}(I-\Pi) (V^B-U^{AB}) \Pi}\\
&\le 3 \epsilon.
\end{align}
Since a partial trace cannot increase the trace distance, we see
\begin{align}
\onenorm{\rho^{AR} - \Tr_B\left(U^{AB} \rho^{ABR} U^{AB\dagger}\right)}
&\le \onenorm{\rho^{AR} - \Tr_B\left( V^B \rho^{ABR} V^{B\dagger}\right) } + 2\epsilon\\
&=\onenorm{\rho^{AR} - \Tr_B\left( V^{B\dagger} V^B \Pi \rho^{ABR} \right) } + 2\epsilon\\
&\le \onenorm{\rho^{AR} - \Tr_B\left(\rho^{ABR} \right) } + 5 \epsilon\\
&= 5 \epsilon \label{eq:rhoTrUrhoU}
\end{align}
where $U^{AB}$ is an arbitrary unitary that preserves the code space.
A completely depolarizing channel over the code space can be written as
\begin{align}
 \int  d U^{AB} ~U^{AB} \rho^{ABR} U^{AB\dagger} = \rho^R \otimes \Pi / \Tr(\Pi)
\end{align}
where the integral is over the group of all logical unitaries.
Using Eq.~\eqref{eq:rhoTrUrhoU},
we see that
\begin{align}
\onenorm{\int d U^{AB}~\left[\rho^{AR} - \Tr_B\left(U^{AB} \rho^{ABR} U^{AB\dagger}\right) \right] }
&\le
\int d U^{AB}~\onenorm{\rho^{AR} - \Tr_B\left(U^{AB} \rho^{ABR} U^{AB\dagger}\right)}\\
&\le 5 \epsilon \,.
\end{align}
Let $\omega^A = \Tr_B \Pi / \Tr(\Pi)$. Then,
\begin{align}
\onenorm{\rho^{AR} - \omega^A \otimes \rho^R}
&=
\onenorm{\int  d U^{AB}~\rho^{AR} - \omega^A \otimes \rho^R}\\
&\le
\onenorm{ \int d U^{AB}~\Tr_B ( U^{AB} \rho^{ABR} U^{AB\dagger} ) - \omega^A \otimes \rho^R } + 5 \epsilon\\
&= 5\epsilon.
\end{align}
\end{proof}

\section{Proof of Theorem~\ref{cor:disent}: Disentangling Unitary}\label{app:D}

\begin{lem}
Let $\Pi$ be a code space on $ABC$, and $R$ be a purifying space.
Let $\cV_{B}^{B'B''} = V \cdot V^\dagger$ denote any isometry channel,
where $B'$ is some auxiliary system
and $B''$ is a copy of $AB$.
Then,
\begin{align}
\inf_{\omega^A} \sup_{\rho^{ABCR}}
\bures( \rho^{ACR}, ~\omega^A \otimes \rho^{CR} )
=
\inf_{\omega^{AB'}, \cV_B^{B'B''}} \sup_{\rho^{ABCR}}
\bures( \cV_B^{B'B''} (\rho^{ABCR}),~ \omega^{AB'} \otimes \rho^{B''CR} )
\label{eq:decouplingUnitary}
\end{align}
where $\rho^{B''CR}$ is the same as $\rho^{ABCR}$
but supported on $B''$ instead of $AB$.
\label{DisentanglingUnitary}
\end{lem}
\begin{proof}[Proof of \ref{DisentanglingUnitary} as a corollary of the information-disturbance theorem]
That $LHS \le RHS$ follows at once because the Bures distance can only decrease
by tracing out $B'B''$.

For the other direction, we need to find an isometry and a state $\omega^{AB'}$
that yields the desired minimax value.
Since $\bures^2$ is convex in its arguments,
the supremum over $\rho^{ABCR}$ always occurs at a pure state,
and thus we can only consider when $\rho^{ABCR}$ is pure.
The information-disturbance theorem (Thm.~\ref{thm:infodist}) provides a recovery map $\cR_B^{AB}$
to fidelity $1-\delta^2$ where $\delta = LHS$.
Let $V=V_B^{ABE}$ be a Stinespring isometry of $\cR_B^{AB}$.
The fidelity being at least $1-\delta^2$ means that
\begin{align}
\abs{\bra{ \omega^{DE} \rho^{ABCR}} V_B^{ABE} \ket{\rho^{DBCR}}}  \ge 1-\delta^2
\end{align}
for some pure state $\omega^{DE}$
where $\rho^{DBCR}$ is the same as $\rho^{ABCR}$
but supported on $DBCR$ instead of $ABCR$.
Interchanging the subsystem label $D$ and $A$, we have
\begin{align}
\abs{\bra{ \omega^{AE} \rho^{DBCR} } V_B^{DBE} \ket{\rho^{ABCR}}}  \ge 1-\delta^2 .
\end{align}
Now we set $B' = E$ and $B''=DB$ to finish the proof.
\end{proof}

Conversely,
the recovery map with respect to the noise $\Tr_A$
in the information-disturbance theorem
can be constructed using the disentangling isometry $\cV_B^{B'B''}$.
Suppose the initial pure state is $\rho^{ABCR}$.
The noise erases $A$,
and we are confronted with $\rho^{BCR}$.
Applying $\cV$, we have, to fidelity $\ge 1-\delta^2$,
$\omega^{B'} \otimes \rho^{B''CR}$.
Discarding $B'$ and mapping $B''$ to $AB$ isomorphically,
the recovery succeeds with fidelity $\ge 1-\delta^2$.

\begin{proof}[Direct Proof of \ref{DisentanglingUnitary}]
That $LHS \le RHS$ follows at once because the Bures distance can only decrease
by tracing out $B'B''$.

To show $\delta := LHS \ge RHS$,
fix an $\omega^A$ that achieves the $LHS$ in Eq.~\eqref{eq:decouplingUnitary}.
\begin{align}
\inf_{\rho^{ABCR}}
\fidelity( \rho^{ACR}, ~\omega^A \otimes \rho^{CR} ) = 1 - \delta^2
\end{align}
Since the fidelity is concave in its arguments,
the infimum occurs at a pure state $\rho^{ABCR}$.
Introduce purification auxiliary systems $D$ and $EF$
where $EF$ is a copy of $AB$.
We apply Uhlmann's theorem with purifications $\omega^{AD}0^B$ of $\omega^A$
and $\rho^{EFCR}$ of $\rho^{CR}$.
Define $S=S^\dagger$ to be the swap operator between $AB$ and $EF$.
Then we have,
\begin{align}
1-\delta^2
&=
\inf_{\text{pure }\rho^{ABCR}}
\sup_{W^{BDEF}}
\Re
\bra{\omega^{AD}0^B \rho^{EFCR}}
W^{BDEF}
\ket{\rho^{ABCR}0^{DEF} } \\
&=
\inf_{\text{pure }\rho^{ABCR}}
\sup_{W^{BDEF}}
\Re
\bra{\omega^{ED}0^{F} \rho^{ABCR}}S^{ABEF}
W^{BDEF}
\ket{\rho^{ABCR}0^{DEF} } \\
&=
\inf_{\rho^{ABC}}
\sup_{W^{BDEF}}
\Re\Tr\left[
S^{ABEF}W^{BDEF}\rho^{ABC}\ket{0^{DEF}} \bra{\omega^{ED}0^{F}}
\right]\,.
\end{align}
In the last expression,
the domain of $\rho^{ABC}$ is the set of all code states,
which is convex.
The domain of $W^{BDEF}$ is the set of all unitaries,
but can be relaxed to the set of all operators $U$ on $BDEF$ of operator norm $\le 1$,
which is convex,
since the supremum occurs at a unitary.
The minimax theorem can then be applied since the objective function is linear in both arguments.

Reversing the manipulation, we have
\begin{align}
1-\delta^2
&=
\sup_{U^{BDEF}}
\inf_{\text{pure }\rho^{ABCR}}
\Re
\bra{\omega^{AD}0^B \rho^{EFCR}}
U^{BDEF}
\ket{\rho^{ABCR}0^{DEF} } .
\end{align}
Here the optimal operator $U$ is not necessarily a unitary.
However, the supremum value can be reached by a unitary
because for any operator $U$ with $\norm{U} \le 1$
there is a unitary $V$ such that
$U^{BDEF} = \bra{0^{D'}} V^{BDD'EF} \ket{0^{D'}}$ due to Lemma~\ref{lem:extensionToUnitary}.
The latter amounts to a larger choice of the purifying space $D$ for $\omega^A$.
Hence, for this choice of $V$, the fidelity between two states on $RHS$ of Eq.~\eqref{eq:decouplingUnitary}
is at least $1-\delta^2$.
This is to say that $RHS \le \delta = LHS$.
Therefore, the claim is proved with $V_B^{B'B''} = V^{BDEF}\ket{0^{DEF}}$
where $B' = BD$ and $B''=EF$.
\end{proof}

\section{Proof of Lemma~\ref{lem:localpert}}
\label{app:E}

\begin{proof}
Let $\cR_B^{AB}$ be a recovery map for $\Tr_A$:
\begin{align}
\sup_{\rho^{ABCR}}\bures\left(
\cR_B^{AB}(\rho^{BCR}),~\rho^{ABCR}
\right)  = \delta.
\end{align}
Let $B_i = A \cap B^{+r}$ and $B_o = C \cap B^{+r}$, so that $B^{+r} = B_i B B_o$.
Also, let $Z = A \setminus B^{+r} = A \setminus B_i$.
We will show that
there exists a recovery channel $\mathcal T = \mathcal T_{B^{+r}}^{ZB^{+r}}$
for the code $U \Pi U^\dagger$ such that
\begin{align}
\sup_{\sigma^{ZB_iBCR} \in~ U\Pi U^\dagger \otimes R} \bures\left(
\mathcal T( \sigma^{B_iBCR}), ~\sigma^{ZB_iBCR}
\right)
\le 2\epsilon + \delta.
\end{align}
Although the domain of $\cR_B^{AB}$ is the set of all states supported on $B$,
we may regard the channel $\cR_B^{AB}$ as
the composition of
the complete depolarization $\mathcal E_A^A$ on $A$,
followed by a channel $\mathcal D_{AB}^{AB}$.
That is, $
\cR_B^{AB}(\sigma^B)
=
\mathcal D_{AB}^{AB} \circ \mathcal E_A^A (\sigma^{AB} )
$
for {\em any} state $\sigma^{AB}$ that partial-traces to $\sigma^B$.
We will use this trick in what follows.

We wish to recover
from the erasure of the ``interior'' $Z = A \setminus B^{+r}$
for the code $U \Pi U^\dagger$.
We consider the following channel $\mathcal S$, which is going to be a recovery map,
\begin{align}
\mathcal S(\sigma^{B_iBCR}) = \mathcal U \circ \mathcal D \circ \mathcal U^\dagger
( I^Z/d_Z \otimes \sigma^{B_iBCR})\,.
\end{align}

We first calculate the fidelity of the recovery by $\mathcal S$.
Let $\sigma^{ZB_iBCR} = \mathcal U (\rho^{ABCR} )$
be an arbitrary code state in $U \Pi U^\dagger$,
possibly entangled with reference system $R$.
The recovery is successful up to error
\begin{align}
&\bures\left(
\mathcal S( \sigma^{B_iBCR} ),~
\sigma^{ZB_iBCR}
\right)\\
&=
\bures\left(
\mathcal U^\dagger \circ \mathcal S \circ \mathcal E_Z^Z \circ \mathcal U (\rho^{ZB_iBCR} ),~
\rho^{ZB_iBCR}
\right)\\
&=
\bures\left(
\mathcal D \mathcal U^\dagger \mathcal E_Z^Z \mathcal U ( \rho^{ZB_iBCR} ),~
\rho^{ZB_iBCR}
\right).
\end{align}
$U$ is locality-preserving,
and the channel $\mathcal U^\dagger \mathcal E_Z^Z \mathcal U$
can be replaced by a channel $\mathcal E'$ on $Z^{+r} = A$
at the cost of additional error at most $\epsilon$.
Hence, we can bound the error as
\begin{align}
\bures\left(
\mathcal S( \sigma^{B_iBCR} ),~
\sigma^{ZB_iBCR}
\right)
&\le
\epsilon
+
\bures\left(
\mathcal D \circ \mathcal E'(\rho^{ABCR}),~ \rho^{ABCR}
\right)\\
&=
\epsilon + \delta
\end{align}
where the last step is because $\Tr_A \mathcal E'(\rho^{ABCR}) = \rho^{BCR}$.

Now we investigate the support of $\mathcal S$.
By assumption on $U$, the composition
$\mathcal U \circ \mathcal D \circ \mathcal U^\dagger$
can be approximated by a channel $\mathcal T'$ supported on
$(AB)^{+r} = ABB_o$ to an error $\epsilon$.
Since the domain of $\mathcal S$ is restricted to the states on $B_i BCR$,
we may say that there exists $\mathcal T_{B_iBB_o}^{ZB_iBB_o}$ that recovers
the erasure of $Z$ for the code $U \Pi U^\dagger$ to error $2\epsilon + \delta$.
\end{proof}

\section{Continuity of mutual information}
\label{sec:continuity-MI}

The Fannes-Audeanaert inequality~\cite{Audenaert2007sharp} reads
\begin{align}
\abs{S(\rho) - S(\sigma)} \le H(t) + t \log (d-1)
\end{align}
for any $d$-dimensional states $\rho$ and $\sigma$,
where $t = \half \onenorm{\rho - \sigma}$, and $H(t) = -t \log t - (1-t)\log(1-t)$
is the binary entropy function.

Alicki and Fannes~\cite{AlickiFannes2004Continuity}
have proved that when $t = \half \onenorm{\rho^{AB} - \sigma^{AB}} < \half$,
\begin{align}
\abs{S(\rho^{AB})- S(\rho^B) - S(\sigma^{AB})+S(\sigma^B)}
\le
2H(2t) + 8t \log d_A
\end{align}
where $d_A$ is the dimension of the subsystem $A$.

Assume that subsystems $A$ and $B$ are nontrivial.
It follows that the difference of the mutual information is bounded as
\begin{align}
\abs{I_\rho(A:B) - I_\sigma(A:B)}
&=
\abs{S(\rho^A) - S(\sigma^A) - S(\rho^{A|B}) - S(\sigma^{A|B})}\\
&\le
H(t) + t \log (d_A-1) + 2H(2t) + 8t \log d_A\\
&\le 9t \log d_A + 6 t \log(1/t) \\
&\le 9 t \log (d_A/t) \,.
\end{align}
Since the mutual information is symmetric, $d_A$ can be replaced by $\min( d_A, d_B) \ge 2$.

\end{document}